\documentclass{acmart}

\usepackage[T1]{fontenc}

\usepackage{subcaption}

\usepackage{tikz}
\usetikzlibrary{backgrounds}

\usepackage{todonotes}
\usepackage{color}
\usepackage{xspace}
\usepackage{pifont}

\DeclareMathOperator{\skel}{skel}

\newcommand{\cA}{\ensuremath{\mathcal{A}}\xspace}
\newcommand{\cB}{\ensuremath{\mathcal{B}}\xspace}

\newcommand{\eDom}[1]{\ensuremath{\gamma^{eq}(#1)}}
\newcommand{\eDomOver}[1]{\ensuremath{\gamma^{dist}(#1)}}
\newcommand{\dom}[1]{\ensuremath{\gamma(#1)}}
\newcommand{\cov}[2]{\ensuremath{cov_{#1}(#2)}}
\newcommand{\mCov}[2]{\ensuremath{max\text{-}cov_{#1}(#2)}}

\renewcommand{\phi}{\varphi}
\renewcommand{\epsilon}{\varepsilon}

\title{K-set agreement bounds in round-based models through combinatorial topology}

\author{Adam Shimi}
\affiliation{%
    \institution{IRIT, University of Toulouse}
    \city{Toulouse}
    \country{France}
}
\email{adam.shimi@irit.fr}

\author{Armando Casta\~{n}eda}
\affiliation{%
    \institution{UNAM}
    \city{Mexico City}
    \country{Mexico}
}
\email{armando.castaneda@im.unam.m}


\begin{document}

\begin{abstract}

  Round-based models are very common message-passing models; combinatorial topology applied
  to distributed computing provides sweeping results like general lower bounds. We combine
  both to study the computability of $k$-set agreement.

  Among all the possible round-based models, we consider oblivious ones, where the
  constraints are given only round per round by a set of allowed graphs.
  And among oblivious models, we focus on closed-above ones,
  that is models where the set of possible graphs contains all graphs with more
  edges than some starting graphs. These capture intuitively
  the underlying structure required by some communication model, like containing a ring.

  We then derive lower bounds and upper bounds in one round for $k$-set agreement,
  such that these bounds are proved using combinatorial topology but stated only
  in terms of graph properties. These bounds extend to multiple rounds when limiting
  our algorithms to be oblivious -- recalling only pairs of processes and initial value,
  not who send what and when.

\end{abstract}

\maketitle

\section{Introduction}
\label{sec:intro}

  \subsection{Motivation}
  \label{subsec:motivation}

    Rounds structure many models of distributed computing: they simplify algorithms,
    capture the distributed equivalent of time complexity~\cite{FraigniaudComplexity},
    and underly many fault-tolerant algorithms, like Paxos~\cite{LamportPaxos}.
    A recent trend, with parallel results by Charron-Bost and Schiper~\cite{CharronBostHO}
    on one hand, and Afek and Gafni~\cite{AfekEquiv} and Raynal and Stainer~\cite{RaynalEquiv}
    on the other hand, is using this concept of round for formalizing
    many different models within a common framework.
    But the techniques used for proving results in these models tend to be ad-hoc,
    very specific to some model or setting. What is required going forward is a
    general approach to proving impossibility results and bounds on round-based
    models.

    Actually, there is at least one example of a general mathematical technique
    used in this context: the characterization of consensus solvability
    through point-set topology by Nowak et al.~\cite{Nowak}.
    We propose what might be seen as an extension to higher dimension of
    this intuition, by applying combinatorial topology (instead of point-set
    topology) to bear on $k$-set agreement (instead of just consensus).

    Combinatorial topology abstracts the reasoning around knowledge and
    indistinguishability behind many impossibility results
    in distributed computing.
    It thus provide generic mathematical tools and methods for deriving such
    results~\cite{HerlihyBook}. Moreover, this approach is the only one that
    managed to prove impossibility results and characterization of solvability
    of the $k$-set agreement~\cite{ChaudhuriKSet}, our focus problem.

    Concretely, we look at closed-above round-models,
    that is models where constraints happens
    round per round, and the set of communication graphs allowed
    is the closure-above of a
    set of graphs. These models capture some safety properties,
    where we require some underlying structure in communication, like having
    an underlying star, ring or tree. This is a strict generalization of the
    models with a fixed communication graph considered by
    Casta{\~n}eda et al.~\cite{Sirocco}.

    For our models, we derive upper bound
    and lower bounds on the $k$ for which $k$-set agreement is solvable.
    And although the proofs of the bounds use combinatorial topology, they
    are stated in terms of variants of the domination number, a well-known
    and used combinatorial number on graphs.

  \subsection{Overview}
  \label{subsec:overview}

    \begin{itemize}
      \item We start by defining closed-above models in Section~\ref{sec:prelim}.
      \item Then we give various upper bounds for $k$-set agreement
        in one round on those models in Section~\ref{sec:upper}. These have the advantage
        of not requiring any combinatorial topology.
      \item Next, we introduce in Section~\ref{sec:topo}
        the combinatorial topology necessary for
        our lower bounds, both the basic definitions and our main
        technical lemma.
      \item We then go to lower bounds on round-based models for
        $k$-set agreement in one round in Section~\ref{sec:lower}. Recall that
        these bounds use combinatorial topology, but are stated in terms of graph
        properties.
      \item Finally, Section~\ref{sec:multiple} generalize both
        upper and lower bounds to the case of multiple rounds.
    \end{itemize}

  \subsection{Related Works}
  \label{subsec:related}

    \paragraph{Round-based models} The idea of using rounds for abstracting
      many different models is classical in message-passing. This
      includes the synchronous adversary models of Afek and Gafni~\cite{AfekEquiv}
      and Raynal and Stainer~\cite{RaynalEquiv}; the Heard-Of model of Charron-Bost
      and Schiper~\cite{CharronBostHO}; and the dynamic networks of
      Kuhn et al~\cite{KuhnComputation}.

      Rounds are also used for building
      a distributed theory of time complexity~\cite{FraigniaudComplexity}
      and for structuring fault-tolerant algorithms like
      Paxos~\cite{LamportPaxos}.

      Previous work on the solvability of consensus and $k$-set agreement
      include the characterization of consensus solvability for oblivious
      round-based models of Coulouma et al.~\cite{CouloumaConsensus},
      the failure-detector-based approach of Jeanneau et al.~\cite{FDKSet},
      and the focus on graceful degradation in algorihtms for $k$-set agreement
      of Biely et al.~\cite{BielyKSet}.

    \paragraph{Combinatorial Topology}
      Combinatorial Topology was first applied to the problem
      of $k$-set agreement in wait-free shared memory by
      Herlihy and Shavit~\cite{HerlihyTopology}, Saks and
      Zaharoglou~\cite{SaksTopology} and
      Borowsky and Gafni~\cite{BorowskyTopology}.

      Beyond these first forays, many other results got proved through
      combinatorial topology. Among others, we can cite
      the lower bounds for renaming by Casta{\~{n}}eda and
      Rajsbaum~\cite{CastanedaRenaming} and
      the derivation of lower-bounds for message-passing
      by Herlihy and Rasjbaum~\cite{HerlihyShellable};
      There is even a result by Alistarh et al.~\cite{AlistarhExtension}
      showing that traditional proof techniques (dubed extension-based proofs)
      cannot prove the impossibility of $k$-set agreement in specific shared-memory
      models, whereas techniques from combinatorial topology can.

      For a full treatment of combinatorial topology applied
      to distributed computing, see Herlihy et al.~\cite{HerlihyBook}.

    \paragraph{Combination of Topology and Round-based models}
      Two papers at least applied topology (combinatorial or not) to
      general round-based models in order to study agreement problems:
      Godard and Perdereau~\cite{GodardKSet} used combinatorial topology
      to study $k$-set agreement in models with omission failures; and
      Nowak et al.~\cite{Nowak} characterization of consensus
      for general round-based models (not necessarily oblivious) using
      point-set topology.

\section{Definitions}
\label{sec:prelim}

  \subsection{Communication models}
  \label{subsec:comModels}

    One common feature of many models of distributed computation is the
    notion of rounds, or layers. Formally, rounds are communication-closed
    as defined by Elrad and Francez~\cite{Elrad}: at each round, a process $p$
    only takes into account the messages (or information) sent by other processes
    at this same round.

    Traditionally, rounds are thought of as synchronous: synchrony indeed provides
    a natural way to implement them. But asynchronous rounds also exist,
    both in message-passing~\cite{CharronBostHO} and in
    shared-memory~\cite{Borowsky,HerlihyAsync}.

    Here we abstract away all implementation details, and consider a model
    with rounds, parameterized with the allowed sequences of communication graphs
    -- directed graphs where each node is a process and each arrow correspond to a
    delivered message to the destination from the source. There are no crashes,
    just a specification of which message can be received at which round.
    This abstracts the Heard-Of model~\cite{CharronBostHO}, synchronous
    message adversaries~\cite{AfekEquiv,RaynalEquiv} and dynamic
    networks~\cite{KuhnComputation,KuhnCoordinated}, as
    well as all other models relying only on the properties of rounds.

    We fix $\Pi = \{p_1,...,p_n\}$ as our set of $n$ processes for the rest of the paper.

    \begin{definition}[Communication model]
        Let $Graphs_{\Pi}$ be the set of graphs. Then
        $Com \subseteq (Graphs_{\Pi})^{\omega}$ is a \textbf{communication model}.
    \end{definition}

    Any set of infinite sequences of graphs defines a model. In order to
    make models more manageable, we focus on a restricted form,
    where the graph for each round is decided independently
    of the others. The model is thus entirely characterized by the set of allowed
    graphs. We call these communication models \textit{oblivious}, following
    Coulouma et al~\cite{CouloumaConsensus}.

    \begin{definition}[Oblivious communication models]
        Let $Com$ be a communication model. Then $Com$ is \textbf{oblivious}
        $\triangleq \exists S \subseteq Graphs_{\Pi}: Com = S^{\omega}$.
    \end{definition}

    Intuitively, oblivous models capture safety properties: bad things that must not
    happen. Or equivalently, good things that must happen at every round. Usually,
    these good properties are related to connectivity, like containing a cycle or
    a spanning tree. Since such a property tends to be invariant
    when more messages are sent, we can look at oblivious models
    defined by a set of subgraphs.

    \begin{definition}[Closed-above communication models]
        Let $Com$ be an oblivious communication model. Then $Com$ is
        \textbf{closed-above} $\triangleq \exists S \subseteq Graphs_{\Pi}:
        Com = (\bigcup\limits_{G \in S} \uparrow G)^{\omega}$, where
        $\uparrow G \triangleq \{ H \mid V(H) = V(G) \land E(H) \supseteq E(G)\}$.

        We call the graphs in $S$ the \textbf{generators} of $Com$.

        If $S$ is a singleton, then $Com$ is \textbf{simple closed-above}.
    \end{definition}

    Classical examples of closed-above models are the non-empty kernel
    predicate (only graphs where at least one process broadcasts) and the
    non-split predicate (only graphs where each pair of processes hears
    from a common process), used notably by Charron-Bost et
    al.~\cite{CharronBostApprox} for characterizing the solvability
    of \textit{approximate consensus} (the variant of consensus where
    the decided value should be less than $\epsilon$ apart, where $\epsilon > 0$
    is fixed beforehand). Another closed-above model is the one satisfying the
    tournament property of Afek and Gafni~\cite{AfekEquiv}, which they show
    is equivalent to wait-free read-write shared memory.

    One example of an oblivous model which is \textbf{not} closed-above
    is the one generated by all graphs containing a cycle, except the clique.
    More generally, the closure-above forces us to have all graphs with more
    edges than our generators.

    Nonetheless, closed-above models capture a fundamental intuition behind
    distributed computing models: specifying what should not happen. They also
    have a good tradeoff between expressivity and simplicity,
    since the "combinatorial data" used to build them is contained in
    a small number of graphs.
    Finally, the patterns expected by safety properties tend to be independent of
    which processes play which roles -- what matters is the existence of a ring or spanning
    tree, not who is where on it.

    We call such closed-above models \textit{symmetric}.

    \begin{definition}[Symmetric models]
        Let $Com$ be a closed-above model, and $S$ be the set of
        graphs generating it. Then $Com$ is \textbf{symmetric}
        $\triangleq S = Sym(S)$, where $Sym(S) = \{ \pi(G) \mid
        G \in S \land \pi : \Pi \to \Pi$ a permutation on $\Pi\}$.
    \end{definition}

    In the rest of the paper, we will limit ourselves to closed-above models, both
    symmetric and not.

  \subsection{Oblivious algorithms}
  \label{subsec:algos}

    Because most applications of combinatorial topology to distributed computing
    aim towards impossibility results, the traditional algorithms considered
    err on the side of power: full information protocols, which exchange at each
    round the view of everything ever heard by the process. For example, after a couple
    of rounds, views will contain nested sets of views, themselves
    containing views, recursively until the initial values.

    In contrast, we focus on oblivious algorithms. That is, we limit each process
    to remember only the initial values it knows, not who sent them or when. This
    amounts to a function from $\Pi$ to the set of initial values
    (with a $\bot$ when the value is not known).
    In turn, these algorithms lose the ability to trace the path of the value.

    We can view oblivious algorithms as full-information
    protocol whose decision map (the function from final view to decision
    value) depends only on the set of known pairs (process,initial value).
    The full-information protocol might still be used for deciding when
    to apply the decision map, but this map loses everything except the
    known pairs. That is, the decision map is constrained to decide similarly
    in situations where it received the same values, even when they were from
    different processes.

    \begin{definition}[Oblivious algorithm]
      Let $\mathcal{A}$ be a full-information protocol, with decision
      map $\delta$. Then $\mathcal{A}$ is an \textbf{oblivious algorithm}
      $\triangleq \forall v$ a view : $\delta(v) = \delta(flat(v))$,\\
      where $flat(v) = \bigcup\limits_{(p,v_p) \in v} flat((p,v_p))$\\
      and $flat((p,v_p)) =
      \left\{
      \begin{array}{ll}
        \{(p,v_p)\} & \text{if } v_p \text{ is a singleton from } V_{in}\\
        flat(v_p) & \text{otherwise}\\
      \end{array}
      \right.$
    \end{definition}

\section{One round upper bounds: a start without topology}
\label{sec:upper}

    Although lower bounds are our targets, they require upper bounds to gauge their
    strength. We thus start with upper bounds on $k$-set agreement~\cite{ChaudhuriKSet}
    for closed-above models.
    Another advantage of starting with our upper bounds is that they rely on
    concrete algorithms, and allow us to introduce generalizations of the classical
    domination number that will be used for our lower bounds.

    Lastly, we also start with bounds for the one round case in this section
    and the next one. Bounds for multiple rounds depend on these one round bounds.

    These bounds follow from a very simple algorithm for solving
    $k$-set agreement. We assume the set of initial values is totally
    ordered. Then everybody sends its initial value for one round,
    and decide the minimum it received.

    \subsection{Simple closed-above models: almost too easy}
    \label{subsec:upperSimple}

        Recall that the domination number of a graph is the size of its smallest
        dominating set, that is the size of the smallest set of nodes whose set
        of outgoing neighbors is $\Pi$. Note that the outgoing neighbors of
        a set $S \subseteq \Pi$ contains $S$ -- that is, we assume self-loop.

        \begin{definition}[Domination number]
            Let $G$ be a graph. Then its \textbf{domination number} $\dom{G}
            \triangleq min \{i \in [1,n] \mid \exists P \subseteq \Pi: |P| = i
            \land \bigcup\limits_{p \in P} Out_G(p) = \Pi \}$.
        \end{definition}

        Because the simple closed-above model generated by $G$ only allows graphs
        containing $G$, their domination number is at most $\dom{G}$. This entails
        a very simple upper bound on $k$-set agreement.

        \begin{theorem}[Upper bound on $k$-set agreement by $\dom{G}$]
            \label{upperDomOne}
            Let $G$ be a graph. Then $\dom{G}$-set agreement
            is solvable in one round on the simple closed-above model generated by $G$.
        \end{theorem}

        \begin{proof}
            The algorithm is just slightly different from the one stated at
            the start of the section: after one round, each process
            decides the minimum value of the ones of a fixed minimum dominating set
            of $G$. Since $G$ is known, this minimum dominating set can be computed
            beforehand. And because it is a dominating set, every process receives
            at least one value from it, so every process can decide.

            Finally, since the minimum dominating set has at most \dom{G} distinct values,
            at most \dom{G} values are decided, and thus our algorithm solves
            \dom{G}-set agreement.
        \end{proof}

        From Casta{\~n}eda et al.~\cite[Thm 5.1]{Sirocco}, we know this bound is
        tight: the oblivious model with a single graph $G$ cannot solve $k$-set agreement
        in one round for $k < \dom{G}$. Hence the weaker simple closed-above model
        generated by $G$ cannot solve $k$-set agreement in one round for
        $k < \dom{G}$.

        Still, simple closed-above models are somewhat artificial, as
        can be seen in the proof: we know exactly the subgraph that must
        be contained in the actual communication graph. A more realistic
        take requires to spread the uncertainty to the underlying subgraph;
        we thus look next at general closed-above models.

    \subsection{General closed-above models: tweaking of upper bounds}
    \label{subsec:upperGeneral}

        For general closed-above models, we must deal with a set
        of possible underlying subgraphs. This makes our previous
        approach inapplicable: we cannot hardcode a dominating
        set because we don't know the underlying subgraph for sure.

        This new issue motivates the definition of a weakening of
        the domination number: the equal-domination number of a set
        of graphs. Intuitively, any set of that much process is a
        dominating set in all the graphs considered.

        \begin{definition}[Equal Domination number of a set of graphs]
            Let $S$ be a set of graphs. Then its \textbf{equal domination
            number} $\eDom{S} \triangleq \max\limits_{G \in S}
            \eDom{G}$, where $\eDom{G} = min \{ i \in [1,n] \mid
            \forall P \subseteq \Pi: |P| = i \implies
            \bigcup\limits_{p \in P} Out_G(p) = \Pi \}$.
        \end{definition}

        \begin{theorem}[Upper bound on $k$-set agreement by \eDom{S}
                        for general closed-above models]
            \label{upperEdomOne}
            Let $S$ be a set of graphs.
            Then \eDom{S}-set agreement is solvable on
            the closed-above model generated by $S$.
        \end{theorem}

        \begin{proof}
            Let $P$ be a set of \eDom{S} processes with the smallest initial values.
            They have thus at most \eDom{S} distinct initial values. By
            definition of \eDom{S}, $P$ dominates every graph in $S$, and
            thus every graph in the closed-above model generated by $S$.

            Thus taking the minimum after one round will result in deciding
            one of those initial values, and thus one of at most \eDom{S}
            values. We conclude that our algorithm solves \eDom{S}-set
            agreement after one round on the closed-above model generated
            by $S$.
        \end{proof}

        Since the equal-domination number is independent of which process does what,
        it is the same for any permutation of the graph. This entails an upper
        bound on symmetric models as a corollary.

        \begin{corollary}
            Let $S$ be a set of graphs.
            Then \eDom{S}-set agreement is solvable on
            the closed-above model generated by $Sym(S)$.
        \end{corollary}

        Now, the natural question to ask is whether we can improve this bound.
        Or equivalently, is it tight?

        The answer depends on the graphs. To see it, let us look at
        another combinatorial number : covering numbers.
        Given fewer processes than the equal-domination number
        of the graph, they do not always form a dominating set. Nonetheless,
        they might still get heard by some minimum number of processes.
        We call such minimums the covering numbers of the graph: the $i$-th covering
        number of $G$ is, given any set of $i$ processes, the minimum number of processes
        hearing this set in $G$.

        \begin{definition}[Covering numbers of a set of graphs]
            Let $S$ be a set of graphs. Then $\forall i < \eDom{S}$,
            its $i$-th covering number $\cov{i}{S} \triangleq \min\limits_{G \in S} \cov{i}{G}$,
            where $\cov{i}{G} \triangleq
            \min\limits_{\substack{P \subseteq \Pi\\ |P| = i}}
            |(\bigcup\limits_{p \in P} Out_G(p))|$.
        \end{definition}

        These numbers capture the ability of a set of processes to disseminate their
        values in the graph. If we take the $i$ processes with the smallest initial values,
        we can be sure that at least $\cov{i}{S}$ processes will hear, and thus
        choose one of these. This then gives a solution to
        $(i + (n-\cov{i}{S}))$-set agreement in one round.

        \begin{theorem}[Upper bounds on $k$-set agreement by covering numbers
                        for general closed-above models]
            \label{upperCovOne}
            Let $S$ be a set of graphs. Then $\forall i \in [1,
            \eDom{S}[: (i+ (n- \cov{i}{S}))$-set agreement is solvable on the oblivious
            closed-above model generated by $S$.
        \end{theorem}

        \begin{proof}
            For a set of $i$ processes with the $i$ smallest initial
            values, they will reach at least $\cov{i}{S}$ processes
            after the first round. Thus these processes
            will decide one of the $i$ values when taking the
            smallest value they received.

            As for the rest of the processes, we can't say anything
            about what they will receive, and thus we consider
            the worst case, where they all decide differently,
            and not one of the $i$ smallest values.
            Then the number of decided values is at most $i+(n-\cov{i}{S})$,
            and the theorem follows.
        \end{proof}

        The covering numbers are also independent of processes names; we thus get
        a similar upper bound on symmetric models as a corollary.

        \begin{corollary}
            Let $S$ be a set of graphs. Then $\forall i \in [1,
            \eDom{S}[: (i+ (n- \cov{i}{S}))$-set agreement is solvable on the oblivious
            closed-above model generated by $Sym(S)$.
        \end{corollary}

        When is this new bound better than the one using the equal-domination number?
        When there is some $i$ such that $n-\cov{i}{S} < \eDom{S}-i$.
        Let us take the symmetric models generated by the two graphs in
        Figure~\ref{exampleCom}.

        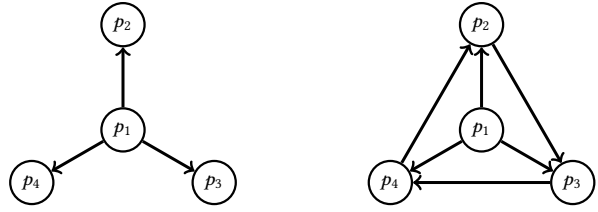
\begin{figure}[t]
        \centering
        \begin{tikzpicture}[scale=.7]
            \tikzstyle{node}=[circle,fill=white,draw,inner sep=.3em, line width=.3mm,font=\footnotesize]
            \tikzstyle{edge}=[->, line width=.4mm]

            \node[node] (p1) at (0,0) {$p_1$};
            \node[node] (p2) at (90:2) {$p_2$};
            \node[node] (p3) at (-30:2) {$p_3$};
            \node[node] (p4) at (210:2) {$p_4$};

            \draw[edge] (p1) to node {} (p2) ;
            \draw[edge] (p1) to node {} (p3) ;
            \draw[edge] (p1) to node {} (p4) ;
        \end{tikzpicture}
        \hspace{5em}
        \begin{tikzpicture}[scale=.7]
            \tikzstyle{node}=[circle,fill=white,draw,inner sep=.3em, line width=.3mm,font=\footnotesize]
            \tikzstyle{edge}=[->, line width=.4mm]

            \node[node] (p1) at (0,0) {$p_1$};
            \node[node] (p2) at (90:2) {$p_2$};
            \node[node] (p3) at (-30:2) {$p_3$};
            \node[node] (p4) at (210:2) {$p_4$};

            \draw[edge] (p1) to node {} (p2) ;
            \draw[edge] (p1) to node {} (p3) ;
            \draw[edge] (p1) to node {} (p4) ;
            \draw[edge] (p2) to node {} (p3) ;
            \draw[edge] (p3) to node {} (p4) ;
            \draw[edge] (p4) to node {} (p2) ;
        \end{tikzpicture}
        \caption{Two examples of communication graphs}\label{exampleCom}
        \end{figure}

        In the first model, $n-\cov{i}{S} < \eDom{S}-i$ never happens,
        because every covering number of a star equals $1$ (the biggest
        set of outgoing neighbors different from $\Pi$ contains only one process),
        and its equal-domination number equals $n$ (because when taking only $n-1$
        processes, the center of the star might not be in there). Thus
        $n-\cov{i}{S} = n-1 \geq \eDom{S}-i=n-i$.

        On the other hand, this is the case in the second model,
        because $\cov{2}{S} = 3$ and $\eDom{S} = 4$.
        We thus we have $n-\cov{2}{S} = 4-3 = 1 < \eDom{S} - i = 4-2 = 2$.
        Hence the upper bound with
        covering numbers ensure $3$-set agreement solvability while the upper bound with
        the equal-domination number only ensures $4$-set agreement solvability.

    \subsection{Intuitions on upper and lower bounds}
    \label{subsec:upperIntuitions}

        Why do our upper bounds hold? Because we can extract from the underlying
        graphs some minimal connectivity of sets of processes. Hence, we know
        from these combinatorial numbers how much the minimal values will spread in the
        worst case, and thus we bound the maximum number of values decided.

        On the other hand, our lower bounds will follow from studying
        how much values can spread in the best case. Why? Because the more values
        can spread, the more processes can distinguish between initial configurations,
        and the more they have a chance to decide correctly. Ensuring enough
        indistinguishability thus entails an impossibility at solving $k$-set agreement.

        This indistinguishability is linked to higher-dimension connectivity
        in combinatorial topology~\cite[Thm. 10.3.1]{HerlihyBook};
        we thus turn to the topological approach
        to distributed computing for our lower bounds.

\section{Elements of combinatorial topology}
\label{sec:topo}

    \subsection{Preliminary definitions}

        First, we need to introduce the mathematical objects that this approach
        uses. These are simplexes and complexes. A simplex is simply a set
        of values, and can be represented as a generalization of a triangle
        in higher dimensions. Simplexes capture configurations in general,
        be them initial configurations, intermediate configurations, or
        decision configurations.

        \begin{definition}[Simplex]
          Let $Cols$ and $Views$ be sets. Then $\sigma \subseteq Cols \times Views$
          is a \textbf{simplex} on $Cols$ and $Views$ (or colored simplex) $\triangleq
          \forall p \in Cols: |\{v \in Views | (p,v) \in \sigma \}| \leq 1$.

          We have $col(\sigma)$ or $names(\sigma) \triangleq
          \{ p \in Cols \mid \exists v \in Views: (p,v) \in \sigma\}$.
          And we have $views(\sigma) =
          \{ v \in Views \mid \exists p \in Cols: (p,v) \in \sigma\}$.
          We also write $view_{\sigma}(p)$ for the $v \in Views$ such that $(p,v) \in \sigma$

          The \textbf{dimension} of $\sigma$ is $|sigma|-1$.
        \end{definition}

        Although we define Views to be any set for readability, the traditional
        view is of sets of pairs, the first element being a process name, and the
        second being either another view or an initial value. For more
        details, refer to~\cite{HerlihyBook}.

        Then a complex is a set of simplexes that is closed under inclusion. It captures
        all considered configurations.

        \begin{definition}[Complex]
          Let $Cols$ and $Views$ be sets. Then $C \in \mathcal{P}(Cols \times Views)$
          is a \textbf{simplicial complex} on $Cols$ and $Views$
          (or colored simplicial complex)
          $\triangleq$
          \begin{itemize}
            \item $\forall (p,v) \in Cols \times Views: \{(p,v)\} \in C$.
            \item $\forall \sigma, \tau$ simplexes on $Cols \times Views$:
              $\sigma \in C \land \tau \subseteq \sigma \implies \tau \in C$.
          \end{itemize}

          The \textbf{facets} of $C$ $\triangleq \{\sigma \in C \mid
          \forall \tau \in C: \sigma \subseteq \tau \implies \tau = \sigma\}$.

          The \textbf{dimension} of $C$ is the maximum dimension of its facets.
          $C$ is called \textbf{pure} if all its facets have the same dimension.
        \end{definition}

        How can we go from our round-based models, which are generated by
        graphs, to simplexes and complexes?

        Starting with a single graph, we define the uninterpreted
        simplex induced by this graph.
        This simplex captures the configuration after a round using graph $G$,
        simply in terms of who hears from whom. It disregards input values,
        which makes it uninterpreted.

        \begin{definition}[Uninterpreted simplex of a graph]
            Let $G$ be a graph. Then the
            \textbf{uninterpreted simplex of $G$} is $\sigma_G \triangleq$
            the colored simplex $\{ (p, In_{G}(p) \mid p \in \Pi \}$.
        \end{definition}

        \begin{figure}
        \centering
        \subcaptionbox{Graph\label{graph}}
        { \begin{tikzpicture}[scale=.6]
            \tikzstyle{node}=[circle,fill=white,draw,inner sep=.3em, line width=.3mm,font=\scriptsize]
            \tikzstyle{edge}=[->, line width=.4mm]

            \node[node, fill=red] (p1) at (0,0) {$p_1$};
            \node[node, fill=green] (p2) at (90:2) {$p_2$};
            \node[node] (p3) at (-30:2) {$p_3$};

            \draw[edge] (p1) to node {} (p2) ;
            \draw[edge] (p3) to node {} (p1) ;
          \end{tikzpicture}}
        \hspace{2cm}
        \subcaptionbox{Uninterpreted simplex\label{simplex}}
        { \begin{tikzpicture}[scale=.6]
            \tikzstyle{node}=[circle,fill=white,draw,inner sep=.3em, line width=.3mm,font=\scriptsize]
            \tikzstyle{edge}=[-, line width=.4mm]

            \node[node, fill=red] (p1) at (0,0) {$(p_1,\{p_1,p_3\}$};
            \node[node, fill=green] (p2) at (5,0) {$(p_2,\{p_1,p_2\}$};
            \node[node] (p3) at (2.5,3) {$(p_3,\{p_3\}$};

            \draw[edge] (p1) to node {} (p2) ;
            \draw[edge] (p1) to node {} (p3) ;
            \draw[edge] (p2) to node {} (p3) ;

            \begin{scope}[on background layer]
                \fill [gray] (p1.center) -- (p2.center) -- (p3.center) -- cycle;
            \end{scope}
          \end{tikzpicture}}
        \caption{A graph and its uninterpreted simplex}\label{exampleUninter}
        \end{figure}
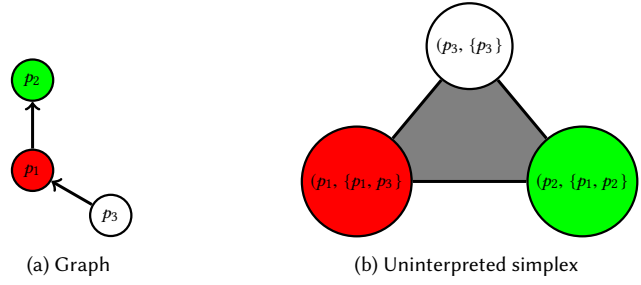

        Given a set of graphs $A$ representing the possible graphs, we generalize
        the previous definition to give the uninterpreted complex of $A$.

        \begin{definition}[Uninterpreted complex of an oblivious model]
            Let $A$ be an oblivous model defined  by a set of graphs $S$.
            Then the \textbf{uninterpreted complex of $A$}
            is $C_A \triangleq$ the complex whose facets are exactly
            the $\{\sigma_G \mid G \in S\}$.
        \end{definition}

    \subsection{Uninterpreted complexes of closed-above models}

        It so happens that closed-above models give rise to uninterpreted
        complex that are easy to define and study. Indeed, they are unions
        of pseudospheres, where pseudospheres are colored complexes topologically
        equivalent to $n$-spheres. These pseudospheres have already been used
        in the literature to study
        multiple models of computation~\cite[Chap. 13]{HerlihyBook}.

        \begin{definition}[Pseudospheres~{\cite[Def 13.3.1]{HerlihyBook}}]
            Let $V_1,V_2,...,V_n$ be sets.
            Then the \textbf{pseudosphere complex} $\phi(\Pi;V_i \mid i \in [1,n])
            \triangleq$
            \begin{itemize}
              \item $\forall i, \forall v \in V_i: (P_i,v)$ is a vertex
                of $C$.
              \item $\forall J \subseteq [1,n]: \{(P_j,v_j) \mid j \in J, v_j \in V_j\}$
                is a simplex of $C$ iff all $P_j$ are distinct.
            \end{itemize}
        \end{definition}

        We can think of these complexes as a generalization of complete
        bipartite graphs in $n$ dimensions. Recall that a complete bipartite graph
        is a graph that can be split into two sets of nodes, the nodes of each set
        not linked to each other and each node of one set linked to all nodes of
        the other set. For example, Figure~\ref{bipart} is a bipartite graph.

        Now a pseudosphere is the same, except that nodes can be partitioned into $n$ sets,
        no simplex contains more than one element of each set as a vertex, and all the simplexes
        built from one element of each set are in the complex.
        Figure~\ref{pseudo} is an example of a pseudosphere built
        from processes $P_1,P_2, P_3$,
        and the three sets $V_1 = \{v_1, v_2\}$, $V_2  = \{v_1, v_2\}$ and $V_3 = \{v\}$.

        \begin{figure}
        \centering
        \subcaptionbox{Bipartite graph\label{bipart}}
        { \begin{tikzpicture}[scale=.8]
            \tikzstyle{node}=[circle,fill=white,draw,inner sep=.3em, line width=.3mm,font=\footnotesize]
            \tikzstyle{edge}=[-, line width=.4mm]

            \node[node] (p1) at (0,0) {$p_1$};
            \node[node] (p2) at (0,1.5) {$p_2$};
            \node[node] (p3) at (0,3) {$p_3$};
            \node[node] (p4) at (2,0) {$p_4$};
            \node[node] (p5) at (2,1.5) {$p_5$};
            \node[node] (p6) at (2,3) {$p_6$};

            \draw[edge] (p1) to node {} (p4) ;
            \draw[edge] (p1) to node {} (p5) ;
            \draw[edge] (p1) to node {} (p6) ;
            \draw[edge] (p2) to node {} (p4) ;
            \draw[edge] (p2) to node {} (p5) ;
            \draw[edge] (p2) to node {} (p6) ;
            \draw[edge] (p3) to node {} (p4) ;
            \draw[edge] (p3) to node {} (p5) ;
            \draw[edge] (p3) to node {} (p6) ;
          \end{tikzpicture}}
        \hspace{2cm}
        \subcaptionbox{Pseudosphere\label{pseudo}}
        { \begin{tikzpicture}[scale=.8]
            \tikzstyle{node}=[circle,fill=white,draw,inner sep=.3em, line width=.3mm,font=\footnotesize]
            \tikzstyle{edge}=[-, line width=.4mm]

            \node[node, fill=red] (p11) at (0,0) {$(P_1,v_1)$};
            \node[node, fill=red] (p12) at (3,3) {$(P_1,v_2)$};
            \node[node, fill=green] (p21) at (0,3) {$(P_2,v_1)$};
            \node[node, fill=green] (p22) at (3,0) {$(P_2, v_2)$};
            \node[node] (p3) at (1.5,1.5) {$(P_3,v)$};

            \draw[edge] (p11) to node {} (p21) ;
            \draw[edge] (p11) to node {} (p22) ;
            \draw[edge] (p11) to node {} (p3) ;
            \draw[edge] (p12) to node {} (p21) ;
            \draw[edge] (p12) to node {} (p22) ;
            \draw[edge] (p12) to node {} (p3) ;
            \draw[edge] (p21) to node {} (p3) ;
            \draw[edge] (p22) to node {} (p3) ;

            \begin{scope}[on background layer]
                \fill [gray] (p11.center) -- (p21.center) -- (p3.center) -- cycle;
                \fill [gray] (p12.center) -- (p21.center) -- (p3.center) -- cycle;
                \fill [gray] (p11.center) -- (p22.center) -- (p3.center) -- cycle;
                \fill [gray] (p12.center) -- (p22.center) -- (p3.center) -- cycle;
            \end{scope}
          \end{tikzpicture}}
        \caption{A bipartite graph and a pseudosphere}
        \end{figure}
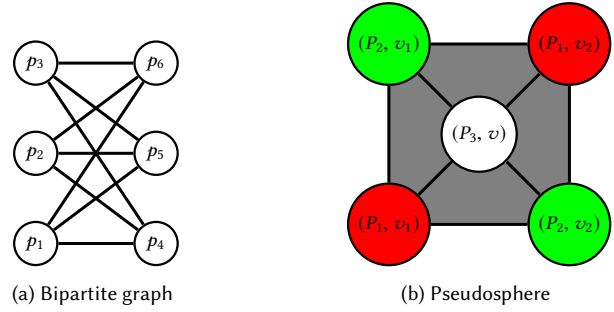

        Among other things, pseudospheres are closed under intersection,
        and are $(n-2)$ connected.

        \begin{lemma}[Intersection of pseudospheres~{\cite[Fact 13.3.4]{HerlihyBook}}]
            \label{capPseudo}
            $\phi(\Pi;U_i \mid i \in [1,n]) \cap \phi(\Pi;V_i \mid i \in [1,n])
            = \phi(\Pi;U_i \cap V_i \mid i \in [1,n])$.
        \end{lemma}

        One advantage of pseudosphere is that they have
        high connectivity~\cite[Def. 3.5.6]{HerlihyBook}. Intuitively,
        connectivity concerns the (non-)existence of high-dimensional generalisation
        of holes in the complexes. Since pseudospheres are topologically equivalent
        to spheres~\cite[Sect. 13.3]{HerlihyBook},
        they only have these holes in the highest dimensions.

        \begin{lemma}[Connectivity of pseudospheres~{\cite[Cor. 13.3.7]{HerlihyBook}}]
            \label{connPseudo}
            $\phi(\Pi;V_i \mid i \in [1,n])$ is $(n-2)$-connected, where
            $n \triangleq |\{i \in [1,n] \mid V_i \neq \emptyset\}|$.
        \end{lemma}

        The connectivity of the uninterpreted complex for a simple closed-above
        model follows, because such a complex is a pseudosphere.
        Intuitively, for any process $p$, its possible views
        are exactly the upward closure of its view in the defining
        graph $G$. Then the $n$-simplexes of the uninterpreted complex
        are exactly the simplex you can build with one such view for each
        process.

        \begin{lemma}[Uninterpreted complex of a simple closed-above
                      model is a pseudosphere]
            \label{simpleClosedPseudo}
            Let $A$ be a simple closed-above model,
            and $G$ be the graph from which it is built.
            Then $C_A =
            \phi(\Pi;\{S \mid In_G(P_i) \subseteq S \subseteq \Pi\}\mid i \in [1,n])$.
        \end{lemma}

        \begin{proof}
            \begin{itemize}
                \item $(\subseteq)$. Let $\sigma$ be a $n$-simplex of $C_A$.
                    By definition of $C_A$, it is the uninterpreted
                    simplex of a graph $H \in \uparrow G$. This in turn
                    means that $\forall p \in \Pi : view_{\sigma}(p) =
                    In_H(p) \supseteq In_G(p)$.

                    Thus $\sigma = \{(P_i,In_H(P_i))\mid i \in [1,n]\} \subseteq
                    \phi(\Pi;\{S \mid In_G(P_i) \subseteq S \subseteq \Pi\}\mid i \in [1,n])$.
                \item $(\supseteq)$. Let $\sigma$ be a n-simplex of
                    $\phi(\Pi;\{S \mid In_G(P_i) \subseteq S \subseteq \Pi\}\mid i \in [1,n])$.
                    Then $\forall p \in \Pi: view_{\sigma}(p) \supseteq In_G(p)$.
                    Thus $\sigma$ is the uninterpreted simplex of
                    a graph $H$ such that $\forall p \in \Pi:
                    In_H(p) \supseteq In_G(p)$.

                    We conclude that $H \in \uparrow G$ and thus that
                    $\sigma \in C_A$
            \end{itemize}
        \end{proof}

        It follows instantly that the uninterpreted complexes of simple closed-above
        models are $(|\Pi|-2)$-connected.

        \begin{corollary}[Connectivity of the uninterpreted complex of
                        a simple closed-above model]
            \label{connSimpleClosed}
            Let $A$ be a simple closed-above model,
            and $G$ be the graph from which it is built.
            Then $C_A$ is $(|\Pi|-2)$-connected.
        \end{corollary}

        From this corollary and the closure of pseudospheres by intersection,
        we now deduce a similar characterization of the connectivity for general closed-above
        models.

        But to do so, we need to first introduce the main tool in our
        toolbox for studying connectivity of simplicial complexes:
        the nerve lemma. This result uses a cover of a complex: a set
        of subcomplexes such that their union gives the initial complex.

        Intuitively, the nerve lemma says that if you provide
        a cover of a complex that is "nice enough", then
        the connectivity of the initial complex can be deduced from
        the way that the cover elements intersects. This is usually
        easier to determine than computing the connectivity directly.

        \begin{definition}[Nerve complex]
            Let $C$ be a simplicial complex, $(C_i)_{i \in I}$ a cover
            of $C$. Then the \textbf{nerve complex} of this cover,
            $\mathcal{N}(C_i \mid I) \triangleq$ the complex generated by
            \begin{itemize}
                \item the vertices are the $C_i$;
                \item and the simplexes are the sets $\{C_i \mid i \in J\}$
                    for $J \subseteq I$ such that
                    $\bigcap\limits_{i \in J} C_i \neq \emptyset$
            \end{itemize}
        \end{definition}

        \begin{lemma}[Nerve lemma~\protect{\cite[Thm 15.24]{Kozlov}}]
            \label{nerveLemma}
            Let $C$ be a simplicial complex, $(C_i)_{i \in I}$ a cover
            of $C$ and $k \geq 0$. Then\\
            $\left(
            \forall J \subseteq I:
            \left(
            \begin{array}{ll}
                dim(\bigcap\limits_{i \in J} C_i) \geq (k-|J|+1)\\
                \bigcap\limits_{i \in J} C_i = \emptyset\\
            \end{array}
            \right)
            \right)\\
            \implies
            (C$ is $k$-connected $\iff \mathcal{N}(C_i \mid I)$ is $k$-connected$)$.
        \end{lemma}

        Now we can prove the connectivity of uninterpreted complexes
        for general closed-above models.

        \begin{theorem}[Connectivity of the uninterpreted complex of
                        a closed-above model]
            \label{connClosed}
            Let $A$ be a closed-above model,
            and $S$ be the set of graphs from which it is built.\\
            Then $C_A$ is $(|\Pi|-2)$-connected.
        \end{theorem}

        \begin{proof}
            From the proof of Theorem~\ref{simpleClosedPseudo}, we know
            that $C_A$ is a union of pseudospheres:
            $C_A = \bigcup\limits_{G \in S} C_G$.
            We want to apply the nerve lemma to this cover. First,
            by Theorem~\ref{connSimpleClosed}, $C_G$ is $(n-2)$-connected.

            As for the intersection of any set $I$ of $C_G$, we have two
            properties. First, it cannot be empty, since all $C_G$ must contains
            the uninterpreted simplex of the complete graph on $\Pi$,
            by definition of $\uparrow G$. This gives us that the nerve complex
            is a simplex, and thus $\infty$-connected.

            And second, the intersection is also a pseudosphere, by application of
            Lemma~\ref{capPseudo}. Indeed, these are intersections
            of pseudospheres with the same processes
            which have an non-empty intersection for each color : the view
            of this process in the complete graph.

            We can thus conclude by application of the nerve lemma and
            Theorem~\ref{connSimpleClosed}.
        \end{proof}

    \subsection{Interpretation of uninterpreted complexes}

        We can only go so far with uninterpreted complexes; at some
        point, we need to consider initial values.

        \begin{definition}[Interpretation of uninterpreted simplex]
            Let $\sigma$ be an uninterpreted simplex on $\Pi$
            and $\tau$ be a $(n-1)$-simplex colored by $\Pi$.
            Then the \textbf{interpretation of $\sigma$ on
            $\tau$}, $\sigma(\tau) \triangleq
            \{(p,V) \mid p \in \Pi \land
            (v \in V \implies (\exists q \in view_{\sigma}(p):
            v = view_{\tau}(q)))\}$
        \end{definition}

        Then the same intuition can be applied to a full uninterpreted
        complex.

        \begin{definition}[Interpretation of uninterpreted complex]
            Let $\mathcal{A}$ be an uninterpreted complex on $\Pi$
            and $\mathcal{I}$ be a pure $(n-1)$ complex colored by
            $\Pi$. Then the \textbf{interpretation of $\mathcal{A}$
            on $\mathcal{I}$}, $\mathcal{A}(\mathcal{I}) \triangleq
            \bigcup\limits_{\substack{\tau \text{ a facet of }\mathcal{I}\\
            \sigma \text{ a facet of } \mathcal{A}}} \sigma(\tau)$
        \end{definition}

        These interpretations give us protocol complexes, on which
        known result on computability are applicable.

    \subsection{A Powerful Tool}
    \label{subsec:tool}

        On the combinatorial topology front, our results leverage two main tools:
        the impossibility result on $k$-set agreement based on
        connectivity~\cite[Thm. 10.3.1]{HerlihyBook},
        and a way to compute the connectivity of a complex from the way it is built. This
        section develops the second idea.

        Let \cA be a pure complex of dimension $d$.
        We say that \cA is \emph{shellable} if there is an ordering
        $\phi_1, \ldots, \phi_r$ of its facets such that for every
        $1 \leq t \leq r-1$,
        $$\left( \bigcup^t_{i=1} \phi_i \right) \cap \phi_{t+1}$$
        is a pure subcomplex of dimension $d-1$ of
        the boundary complex of $\phi_{t+1}$, i.e., of $\skel^{d-1}\phi_{t+1}$.

        The intuition here is that the complex is the union of simplexes of dimension $d$,
        and there is an order in which to add simplexes, so that the new simplex
        is connected to the rest by $d-1$ simplexes, some of its own facets.
        In the concrete case of $2$-simplexes (triangles), they must be connected
        to the rest by $1$-simplexes (edges).

        Here, unions and intersections apply to the complexes induced by the facet and
        all its faces. Such a sequence of facets is a \emph{shelling order} of \cA.

        For example, the complex in Figure~\ref{exShell} is shellable, but the one in
        Figure~\ref{exNotShell} is not.

        \begin{figure}
        \centering
        \subcaptionbox{Shellable complex\label{exShell}}{
        \begin{tikzpicture}[scale=.8]
            \tikzstyle{node}=[circle,fill=white,draw,inner sep=.3em, line width=.3mm,font=\footnotesize]
            \tikzstyle{edge}=[-, line width=.4mm]

            \node[node] (p1) at (0,0) {$p_1$};
            \node[node] (p2) at (1.5,1.5) {$p_2$};
            \node[node] (p3) at (1.5,-1.5) {$p_3$};
            \node[node] (p4) at (3,0) {$p_4$};

            \draw[edge] (p1) to node {} (p2) ;
            \draw[edge] (p1) to node {} (p3) ;
            \draw[edge] (p2) to node {} (p3) ;
            \draw[edge] (p2) to node {} (p4) ;
            \draw[edge] (p4) to node {} (p3) ;

            \begin{scope}[on background layer]
                \fill [gray] (p1.center) -- (p2.center) -- (p3.center) -- cycle;
                \fill [gray] (p4.center) -- (p2.center) -- (p3.center) -- cycle;
            \end{scope}
        \end{tikzpicture}}
        \hspace{5em}
        \subcaptionbox{Not-shellable complex\label{exNotShell}}{
        \begin{tikzpicture}[scale=.8]
            \tikzstyle{node}=[circle,fill=white,draw,inner sep=.3em, line width=.3mm,font=\footnotesize]
            \tikzstyle{edge}=[-, line width=.4mm]

            \node[node] (p1) at (0,0) {$p_1$};
            \node[node] (p2) at (0,3) {$p_2$};
            \node[node] (p3) at (1.5,1.5) {$p_3$};
            \node[node] (p4) at (3,0) {$p_4$};
            \node[node] (p5) at (3,3) {$p_5$};

            \draw[edge] (p1) to node {} (p2) ;
            \draw[edge] (p1) to node {} (p3) ;
            \draw[edge] (p2) to node {} (p3) ;
            \draw[edge] (p4) to node {} (p5) ;
            \draw[edge] (p4) to node {} (p3) ;
            \draw[edge] (p5) to node {} (p3) ;

            \begin{scope}[on background layer]
                \fill [gray] (p1.center) -- (p2.center) -- (p3.center) -- cycle;
                \fill [gray] (p4.center) -- (p5.center) -- (p3.center) -- cycle;
            \end{scope}
        \end{tikzpicture}}
        \caption{Examples of a complex that is shellable and one that is not}
        \end{figure}
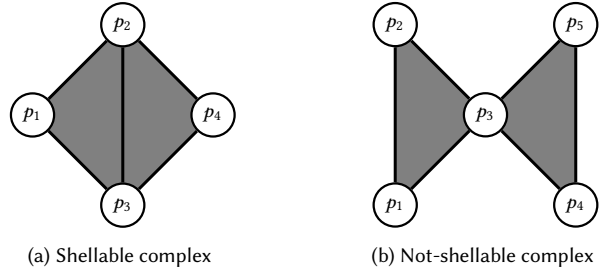

        Given a shelling order $\phi_1, \ldots, \phi_r$ of a complex $\cA$,
        $(\bigcup^t_{i=1} \phi_i) \cap \phi_{t+1}$ is the union of the complexes induced
        by some $(d-1)$-faces $\tau_1, \hdots, \tau_s$ of $\phi_{t+1}$, by definition of shellability.
        Each $\tau_j$ is a face of a facet $\sigma_j$ of $\cup^t_{i=1} \phi_i$,
        hence $\phi_{t+1}$ and $\sigma_j$ share a $(d-1)$-face.
        Then,
        $$\left( \bigcup^t_{i=1} \phi_i \right) \cap \phi_{t+1} = \bigcup^s_{j=1} (\phi_{t+1} \cap \sigma_j).$$

        We also use the following technical result.

        \begin{lemma}[Shellability of simplex boundary~\protect{\cite[Thm 13.2.2]{HerlihyBook}}]
            \label{lemma-shelling-boundary}
            Let \cA be a pure $(d-1)$-dimensional sub-complex of the boundary complex
            of a simplex of dimension $d$.
            Then \cA is shellable, and any sequence of its facets is a shelling order for \cA.
        \end{lemma}

        Finally, we rely on the straightforward corollary of the nerve lemma
        for a cover with two elements.

        \begin{corollary}[Two elements nerve lemma]
          \label{simpleNerve}
          Let $C$ and $K$ be $k$-connected complexes. If
          $C \cap K$ is $(k-1)$-connected, then $C \cup K$ is $k$-connected.
        \end{corollary}

        We now state the main technical result of the section. It extends the
        result from Casta{\~n}eda et al.~\cite{Sirocco}
        and adapts it to the interpretation
        of complexes we need here. While Casta{\~n}eda et al. studied the complex given
        by the interpretation of a single graph (to capture models like
        $\mathcal{LOCAL}$ and $\mathcal{CONGEST}$~\cite{PelegBook}),
        we care about the complex resulting of the interpretation of a set of graphs.

        We thus send each input simplex into a complex, and show that if both the
        output complexes and the mapping are "nice", the interpreted complex is
        highly connected.

        \begin{lemma}
            \label{lemma-shellability-connectivity}
            Let $\mathcal{A}$ be a pure shellable complex of dimension $d$,
            $\mathcal{B}$ a complex, $(\mathcal{B}_i)_{i \in I}$
            a cover of $\mathcal{B}$, and $\ell \geq 0$ an integer.
            Suppose that there is a bijection $\alpha$ between the facets
            of $\mathcal{A}$ and the elements of $(\mathcal{B}_i)_{i \in I}$
            such that:
            \begin{enumerate}
                \item For every facet $\phi'$ of $\cA$ and
                    every pure $d$-subcomplex $\bigcup^t_{i=1} \phi_i
                    \subseteq \cA$ satisfying that\\
                    $\left( \bigcup^t_{i=1} \phi_i \right) \cap \phi' =
                    \bigcup^s_{i=1} \left( \phi' \cap \sigma_i \right)$
                    for some of $\cA$'s facets $\sigma_1, \hdots, \sigma_s$,
                    with each $\sigma_i$ and $\phi'$ sharing a $(d-1)$-face,
                    it holds that $\left( \bigcup^t_{i=1} \alpha(\phi_i) \right)
                    \cap \alpha(\phi') = \bigcup^s_{i=1} \left( \alpha(\phi')
                    \cap \alpha(\sigma_i) \right)$.
                \item For every $t\geq 0$ and every collection
                    $\phi_0, \phi_1, \ldots, \phi_t$ of $t+1$ facets of $\mathcal{A}$
                    with each $\phi_i$ and $\phi_0$ sharing a $(d-1)$-face,
                    it holds that $\bigcap^t_{i=0} \alpha(\phi_i)$
                    is least $(\ell - t)$-connected.
            \end{enumerate}

            Then, $\mathcal{B}$ is $\ell$-connected.
        \end{lemma}

\section{One round lower bounds for one round: a touch of topology}
\label{sec:lower}

  As before, we start with the simple closed-above case, where the model
  is the closure of a single graph. In this case the tight lower bound
  follows from Casta{\~n}eda et al.~\cite[Thm 5.1]{Sirocco}, as mentionned above.

  \begin{theorem}[Lower bound on $k$-set agreement for simple closed-above
                  models]
    Let $A$ a simple closed-above model generated by the graph
    $G$. Let $k \leq \dom{G}$.
    Then $k$-set agreement is not solvable on $A$ in a single round.
  \end{theorem}

  We thus focus on general closed-above models. Here we have to
  leverage the underlying structure of the protocol complex. We do so through
  two tools : the main theorem from Section~\ref{subsec:tool}, as well as
  two graph parameters: the equal-domination number over a set of graphs,
  and the max-covering numbers of a set of graphs.

  \begin{definition}[Distributed domination number of a set of graphs]
    Let $S$ be a set of graphs. Then the \textbf{distributed domination
    number} $S$, $\eDomOver{S} \triangleq
    min \left\{ i > 0 \middle|
    \begin{array}{l}
      \forall P \subseteq \Pi, \forall S_i \subseteq S:\\
      (|P| = i \land |S_i| = \min(i,|S|))\\
      \implies \bigcup\limits_{G \in S_i} Out_G(P) = \Pi\\
    \end{array}
    \right\}$.
  \end{definition}

  The difference between \eDom{S} and \eDomOver{S} is that a set of
  \eDom{S} processes dominates each graph of $S$ separately, whereas
  a set of \eDomOver{S} processes might not dominate any graph of $S$,
  but it dominates every subset of $i$ graphs of $S$ together.
  Thus \eDomOver{S} $\leq$ \eDom{S}. Fitting, considering the former is used
  in lower bounds and the latter in upper bounds.

  Next, the max-covering numbers are quite subtle. For $i < \eDomOver{S}$,
  the $i$-th max-covering number of $S$ is the maximum number of processes
  hearing a set of $i$ procs, summed over $i$ graphs in $S$.

  That is, the max-covering numbers capture how much values can be disseminated
  in the best case. They serve in lower bounds by giving a best case scenario
  on which we can focus to prove impossibility.

  \begin{definition}[Max covering numbers of a set of graphs]
    Let $S$ be a set of graphs and $i < \eDomOver{S}$. Then
    the $i$-th max-covering number of $S$,
    $\textit{max-cov}_i(S) \triangleq
    \max\limits_{\substack{P \subseteq \Pi, |P| = i\\
    S_i \subseteq S, |S_i| = \min(i,|S|)\\
    \bigcup\limits_{G \in S_i} Out_G(P) \neq \Pi}}
    |(\bigcup\limits_{G \in S_i} Out_G(P))|$.

    We also define the $i$-th max-covering coefficients on $S$, $M_i(S)
    \triangleq
    \left\{
    \begin{array}{ll}
      \left\lfloor \frac{n-i-1}{\mCov{i}{S}-i} \right\rfloor & \text{ if } \mCov{i}{S} > i\\
      n-i & \text{ if } \mCov{i}{S} = i\\
    \end{array}
    \right.$
  \end{definition}

  \begin{theorem}[Lower bound on $k$-set agreement for general closed-above
                  models]
    \label{lowerGeneral}
    Let $A$ a closed-above model generated by the set of graphs
    $S$.\\
    Let $l = \min(\eDomOver{S}-2,\min \{t+M_t(S)-2 \mid t \in [1,\eDomOver{S}-1] \})$
    Then $(l+1)$-set agreement is not solvable on $A$
    in a single round.
  \end{theorem}

  The term depending on \eDomOver{S} in the lower bound serves when the max-covering
  numbers are not sufficient to distinguish adversaries with different properties.
  Consider for example the symmetric models of all unions of $s$ stars, with
  $s \leq n$. Then for those graphs, for $t < \eDomOver{S}$, we have
  \mCov{t}{S} = t, and thus $M_t(S) = n-t$. Hence the minimum over the
  $t+M_t(S)-2$ is $n-2$.

  But this would mean that $(n-1)$-set agreement is impossible for $s < n$,
  whereas we can clearly solve $2$-set agreement for $s = n-1$, for example.
  What depends on $s$ is $\eDomOver{S}$ itself. More precisely, $\eDomOver{S}
  = n-s+1$, because given $P$, we can consider only the graph where the $s$
  centers of stars are in $\Pi \setminus P$, up until the point where $|P| > n-s$.

  Hence our lower bound shows that for the symmetric union of $s$ stars, $(n-s)$-set
  agreement is impossible in one round. Given that our upper bounds above tell that
  $(n-s+1)$-set agreement is possible in one round for this model,
  the bound is tight.

  Finally, the bound can be specialized for symmetric models.

  \begin{corollary}[Lower bounds for symmetric closed-above model]
    \label{lowerSym}
    Let $G$ be a graph. Let $l=\\
    \min( \eDomOver{Sym(G)}-2, \min\limits_{t \in [1,\eDomOver{Sym(G)}-1]}\\
    \left(
    \begin{array}{ll}
      \hspace{-.2cm}t +
        \left\lfloor \frac{n-t-1}{t(\mCov{t}{\{G\}}-t)} \right\rfloor - 2
        & \text{ if } \mCov{t}{\{G\}} > t\\
      \hspace{-.2cm}n - 2
        & \text{ if } \mCov{t}{\{G\}} = t\\
    \end{array}
    \right)$\\
    Then $(l+1)$-set agreement is not solvable on $Sym(\uparrow G)$
    in a single round.
  \end{corollary}

  Notice that all these lower bounds are valid for general algorithms, not
  only oblivious ones. The reason is that a one round full information protocol
  is an oblivious algorithm.

\section{Multiple rounds}
\label{sec:multiple}

  Given that we focus on oblivious algorithms, a natural approach to extending
  our lower bounds to the multiple rounds case is to look at the product of our
  graphs. By product, we mean the graph of the paths with one edge per graph.
  Thus the products of $r$ graphs capture who will hear who after $r$ corresponding
  communication rounds.

  \begin{definition}[Graph path product]
    Let $G$ and $H$ be graphs with auto-loops
    ($\forall v \in \Pi: (v,v) \in E(G) \land (v,v) \in E(H)$).
    Then their \textbf{graph path product}
    $G \bigotimes H \triangleq$ the graph $(\Pi,E)$ such that $\forall u,v \in \Pi:
    (u,v) \in \Pi \implies \exists w \in \Pi: (u,w) \in E(G) \land (w,v) \in E(H)$.
  \end{definition}

  Since we have a graph as the result, we can apply our lower bounds for one round.
  At least, if the resulting graph still satisfy the hypotheses of our lower bounds.
  It does, although product doesn't maintain closure-above. This subtlety is
  explained in the next subsection.

  \subsection{Closure-above is not invariant by product, but its still works}

    What is the pitfall mentionned above? Quite simply, that the product of
    two closed-above models does not necessarily gives a closed-above model.
    This follows from the fact that the closure-above of a product of graphs
    doesn't always equal the product of the closure-above of the graphs.

    \tikzset{ring/.pic={
      \tikzstyle{node}=[circle,fill=white,draw,inner sep=.3em, line width=.3mm,font=\footnotesize]
      \tikzstyle{edge}=[->, line width=.4mm, bend left=20]

      \node[node] (p1) at (2*360/6:2) {$p_1$};
      \node[node] (p2) at (360/6:2) {$p_2$};
      \node[node] (p3) at (0:2) {$p_3$};
      \node[node] (p4) at (5*360/6:2) {$p_4$};
      \node[node] (p5) at (4*360/6:2) {$p_5$};
      \node[node] (p6) at (3*360/6:2) {$p_6$};

      \draw[edge] (p1) to node {} (p2) ;
      \draw[edge] (p2) to node {} (p3) ;
      \draw[edge] (p3) to node {} (p4) ;
      \draw[edge] (p4) to node {} (p5) ;
      \draw[edge] (p5) to node {} (p6) ;
      \draw[edge] (p6) to node {} (p1) ;
    }}

    \tikzset{ringPlusBase/.pic={
      \tikzstyle{node}=[circle,fill=white,draw,inner sep=.3em, line width=.3mm,font=\footnotesize]
      \tikzstyle{edge}=[->, line width=.4mm, bend left=20]

      \node[node] (p1) at (2*360/6:2) {$p_1$};
      \node[node] (p2) at (360/6:2) {$p_2$};
      \node[node] (p3) at (0:2) {$p_3$};
      \node[node] (p4) at (5*360/6:2) {$p_4$};
      \node[node] (p5) at (4*360/6:2) {$p_5$};
      \node[node] (p6) at (3*360/6:2) {$p_6$};

      \draw[edge] (p1) to node {} (p2) ;
      \draw[edge] (p2) to node {} (p3) ;
      \draw[edge] (p3) to node {} (p4) ;
      \draw[edge] (p4) to node {} (p5) ;
      \draw[edge] (p5) to node {} (p6) ;
      \draw[edge] (p6) to node {} (p1) ;

      \draw[edge,green] (p2) to node {} (p6) ;
    }}

    \tikzset{ringPlusCon/.pic={
      \tikzstyle{node}=[circle,fill=white,draw,inner sep=.3em, line width=.3mm,font=\footnotesize]
      \tikzstyle{edge}=[->, line width=.4mm, bend left=20]

      \node[node] (p1) at (2*360/6:2) {$p_1$};
      \node[node] (p2) at (360/6:2) {$p_2$};
      \node[node] (p3) at (0:2) {$p_3$};
      \node[node] (p4) at (5*360/6:2) {$p_4$};
      \node[node] (p5) at (4*360/6:2) {$p_5$};
      \node[node] (p6) at (3*360/6:2) {$p_6$};

      \draw[edge] (p1) to node {} (p2) ;
      \draw[edge] (p2) to node {} (p3) ;
      \draw[edge] (p3) to node {} (p4) ;
      \draw[edge] (p4) to node {} (p5) ;
      \draw[edge] (p5) to node {} (p6) ;
      \draw[edge] (p6) to node {} (p1) ;

      \draw[edge,green] (p3) to node {} (p6) ;
    }}

    \tikzset{ringPlusChaud1/.pic={
      \tikzstyle{node}=[circle,fill=white,draw,inner sep=.3em, line width=.3mm,font=\footnotesize]
      \tikzstyle{edge}=[->, line width=.4mm, bend left=20]

      \node[node] (p1) at (2*360/6:2) {$p_1$};
      \node[node] (p2) at (360/6:2) {$p_2$};
      \node[node] (p3) at (0:2) {$p_3$};
      \node[node] (p4) at (5*360/6:2) {$p_4$};
      \node[node] (p5) at (4*360/6:2) {$p_5$};
      \node[node] (p6) at (3*360/6:2) {$p_6$};

      \draw[edge] (p1) to node {} (p2) ;
      \draw[edge] (p2) to node {} (p3) ;
      \draw[edge] (p3) to node {} (p4) ;
      \draw[edge] (p4) to node {} (p5) ;
      \draw[edge] (p5) to node {} (p6) ;
      \draw[edge] (p6) to node {} (p1) ;

      \draw[edge,green] (p2) to node {} (p4) ;
    }}

    \tikzset{ringPlusChaud2/.pic={
      \tikzstyle{node}=[circle,fill=white,draw,inner sep=.3em, line width=.3mm,font=\footnotesize]
      \tikzstyle{edge}=[->, line width=.4mm, bend left=20]

      \node[node] (p1) at (2*360/6:2) {$p_1$};
      \node[node] (p2) at (360/6:2) {$p_2$};
      \node[node] (p3) at (0:2) {$p_3$};
      \node[node] (p4) at (5*360/6:2) {$p_4$};
      \node[node] (p5) at (4*360/6:2) {$p_5$};
      \node[node] (p6) at (3*360/6:2) {$p_6$};

      \draw[edge] (p1) to node {} (p2) ;
      \draw[edge] (p2) to node {} (p3) ;
      \draw[edge] (p3) to node {} (p4) ;
      \draw[edge] (p4) to node {} (p5) ;
      \draw[edge] (p5) to node {} (p6) ;
      \draw[edge] (p6) to node {} (p1) ;

      \draw[edge,green] (p4) to node {} (p6) ;
    }}

    \tikzset{ringSquared/.pic={
      \tikzstyle{node}=[circle,fill=white,draw,inner sep=.3em, line width=.3mm,font=\footnotesize]
      \tikzstyle{edge}=[->, line width=.4mm, bend left=20]
      \tikzstyle{flat}=[->, line width=.4mm]

      \node[node] (p1) at (2*360/6:2) {$p_1$};
      \node[node] (p2) at (360/6:2) {$p_2$};
      \node[node] (p3) at (0:2) {$p_3$};
      \node[node] (p4) at (5*360/6:2) {$p_4$};
      \node[node] (p5) at (4*360/6:2) {$p_5$};
      \node[node] (p6) at (3*360/6:2) {$p_6$};

      \draw[edge] (p1) to node {} (p2) ;
      \draw[edge] (p2) to node {} (p3) ;
      \draw[edge] (p3) to node {} (p4) ;
      \draw[edge] (p4) to node {} (p5) ;
      \draw[edge] (p5) to node {} (p6) ;
      \draw[edge] (p6) to node {} (p1) ;

      \draw[flat] (p1) to node {} (p3) ;
      \draw[flat] (p2) to node {} (p4) ;
      \draw[flat] (p3) to node {} (p5) ;
      \draw[flat] (p4) to node {} (p6) ;
      \draw[flat] (p5) to node {} (p1) ;
      \draw[flat] (p6) to node {} (p2) ;
    }}

    \tikzset{ringSquaredPlus/.pic={
      \tikzstyle{node}=[circle,fill=white,draw,inner sep=.3em, line width=.3mm,font=\footnotesize]
      \tikzstyle{edge}=[->, line width=.4mm, bend left=20]

      \node[node] (p1) at (2*360/6:2) {$p_1$};
      \node[node] (p2) at (360/6:2) {$p_2$};
      \node[node] (p3) at (0:2) {$p_3$};
      \node[node] (p4) at (5*360/6:2) {$p_4$};
      \node[node] (p5) at (4*360/6:2) {$p_5$};
      \node[node] (p6) at (3*360/6:2) {$p_6$};

      \draw[edge] (p1) to node {} (p2) ;
      \draw[edge] (p2) to node {} (p3) ;
      \draw[edge] (p3) to node {} (p4) ;
      \draw[edge] (p4) to node {} (p5) ;
      \draw[edge] (p5) to node {} (p6) ;
      \draw[edge] (p6) to node {} (p1) ;

      \draw[edge] (p1) to node {} (p3) ;
      \draw[edge] (p2) to node {} (p4) ;
      \draw[edge] (p3) to node {} (p5) ;
      \draw[edge] (p4) to node {} (p6) ;
      \draw[edge] (p5) to node {} (p1) ;
      \draw[edge] (p6) to node {} (p2) ;

      \draw[edge, green] (p2) to node {} (p6) ;
    }}

    \tikzset{ringSquaredPlusFail1/.pic={
      \tikzstyle{node}=[circle,fill=white,draw,inner sep=.3em, line width=.3mm,font=\footnotesize]
      \tikzstyle{edge}=[->, line width=.4mm, bend left=20]

      \node[node] (p1) at (2*360/6:2) {$p_1$};
      \node[node] (p2) at (360/6:2) {$p_2$};
      \node[node] (p3) at (0:2) {$p_3$};
      \node[node] (p4) at (5*360/6:2) {$p_4$};
      \node[node] (p5) at (4*360/6:2) {$p_5$};
      \node[node] (p6) at (3*360/6:2) {$p_6$};

      \draw[edge] (p1) to node {} (p2) ;
      \draw[edge] (p2) to node {} (p3) ;
      \draw[edge] (p3) to node {} (p4) ;
      \draw[edge] (p4) to node {} (p5) ;
      \draw[edge] (p5) to node {} (p6) ;
      \draw[edge] (p6) to node {} (p1) ;

      \draw[edge] (p1) to node {} (p3) ;
      \draw[edge] (p2) to node {} (p4) ;
      \draw[edge] (p3) to node {} (p5) ;
      \draw[edge] (p4) to node {} (p6) ;
      \draw[edge] (p5) to node {} (p1) ;
      \draw[edge] (p6) to node {} (p2) ;

      \draw[edge, green] (p2) to node {} (p6) ;
      \draw[edge, red] (p2) to node {} (p1) ;
    }}

    \tikzset{ringSquaredPlusFail2/.pic={
      \tikzstyle{node}=[circle,fill=white,draw,inner sep=.3em, line width=.3mm,font=\footnotesize]
      \tikzstyle{edge}=[->, line width=.4mm, bend left=20]

      \node[node] (p1) at (2*360/6:2) {$p_1$};
      \node[node] (p2) at (360/6:2) {$p_2$};
      \node[node] (p3) at (0:2) {$p_3$};
      \node[node] (p4) at (5*360/6:2) {$p_4$};
      \node[node] (p5) at (4*360/6:2) {$p_5$};
      \node[node] (p6) at (3*360/6:2) {$p_6$};

      \draw[edge] (p1) to node {} (p2) ;
      \draw[edge] (p2) to node {} (p3) ;
      \draw[edge] (p3) to node {} (p4) ;
      \draw[edge] (p4) to node {} (p5) ;
      \draw[edge] (p5) to node {} (p6) ;
      \draw[edge] (p6) to node {} (p1) ;

      \draw[edge] (p1) to node {} (p3) ;
      \draw[edge] (p2) to node {} (p4) ;
      \draw[edge] (p3) to node {} (p5) ;
      \draw[edge] (p4) to node {} (p6) ;
      \draw[edge] (p5) to node {} (p1) ;
      \draw[edge] (p6) to node {} (p2) ;

      \draw[edge, green] (p2) to node {} (p6) ;

      \draw[edge, red] (p1) to node {} (p6) ;
    }}

    \tikzset{ringSquaredPlusFail3/.pic={
      \tikzstyle{node}=[circle,fill=white,draw,inner sep=.3em, line width=.3mm,font=\footnotesize]
      \tikzstyle{edge}=[->, line width=.4mm, bend left=20]

      \node[node] (p1) at (2*360/6:2) {$p_1$};
      \node[node] (p2) at (360/6:2) {$p_2$};
      \node[node] (p3) at (0:2) {$p_3$};
      \node[node] (p4) at (5*360/6:2) {$p_4$};
      \node[node] (p5) at (4*360/6:2) {$p_5$};
      \node[node] (p6) at (3*360/6:2) {$p_6$};

      \draw[edge] (p1) to node {} (p2) ;
      \draw[edge] (p2) to node {} (p3) ;
      \draw[edge] (p3) to node {} (p4) ;
      \draw[edge] (p4) to node {} (p5) ;
      \draw[edge] (p5) to node {} (p6) ;
      \draw[edge] (p6) to node {} (p1) ;

      \draw[edge] (p1) to node {} (p3) ;
      \draw[edge] (p2) to node {} (p4) ;
      \draw[edge] (p3) to node {} (p5) ;
      \draw[edge] (p4) to node {} (p6) ;
      \draw[edge] (p5) to node {} (p1) ;
      \draw[edge] (p6) to node {} (p2) ;

      \draw[edge, green] (p2) to node {} (p6) ;

      \draw[edge, red] (p2) to node {} (p5) ;
      \draw[edge, red] (p3) to node {} (p6) ;
    }}

    \tikzset{ringSquaredPlusFail4/.pic={
      \tikzstyle{node}=[circle,fill=white,draw,inner sep=.3em, line width=.3mm,font=\footnotesize]
      \tikzstyle{edge}=[->, line width=.4mm, bend left=20]

      \node[node] (p1) at (2*360/6:2) {$p_1$};
      \node[node] (p2) at (360/6:2) {$p_2$};
      \node[node] (p3) at (0:2) {$p_3$};
      \node[node] (p4) at (5*360/6:2) {$p_4$};
      \node[node] (p5) at (4*360/6:2) {$p_5$};
      \node[node] (p6) at (3*360/6:2) {$p_6$};

      \draw[edge] (p1) to node {} (p2) ;
      \draw[edge] (p2) to node {} (p3) ;
      \draw[edge] (p3) to node {} (p4) ;
      \draw[edge] (p4) to node {} (p5) ;
      \draw[edge] (p5) to node {} (p6) ;
      \draw[edge] (p6) to node {} (p1) ;

      \draw[edge] (p1) to node {} (p3) ;
      \draw[edge] (p2) to node {} (p4) ;
      \draw[edge] (p3) to node {} (p5) ;
      \draw[edge] (p4) to node {} (p6) ;
      \draw[edge] (p5) to node {} (p1) ;
      \draw[edge] (p6) to node {} (p2) ;

      \draw[edge, green] (p2) to node {} (p6) ;
      \draw[edge, red] (p3) to node {} (p6) ;
    }}

    Let's take an example: the product of a cycle with itself.

    
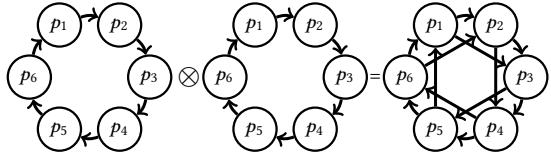
\begin{center}
    \begin{tikzpicture}[scale=.4]
      \pic[scale=.4] at (0,0) {ring};
      \node[scale=.4] (op) at (3.3,0) {\Huge $\bigotimes$};
      \pic[scale=.4] at (6.5,0) {ring};
      \node[scale=.4] (eq) at (9.5,0) {\Huge $=$};
      \pic[scale=.4] at (12.5,0) {ringSquared};
    \end{tikzpicture}
    \end{center}

    Then we cannot build the following graph by extending the cycles and taking the
    product:

    
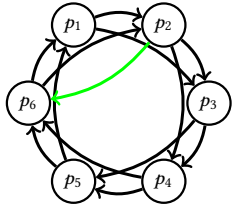
\begin{center}
    \begin{tikzpicture}
      \pic[scale=.6] {ringSquaredPlus};
    \end{tikzpicture}
    \end{center}

    Why? Simply put, adding the new edge to either of the two cycles necessarily
    creates other edges in the product. Adding an edge from $p_2$ to any other
    node than $p_3$ and $p_4$ also creates new edges; so does adding an edge
    to $p_4$ and then an edge from $p_4$ to $p_6$, or an edge from $p_3$ to $p_6$ in
    the second graph.

    Hence the product of the closure above of this cycle with itself is not
    the closure-above of the squared cycle. To put it differently, closure-above
    is not invariant by the product operation.

    Nonetheless, the bell does not toll for our hopes of extending our properties.
    What is used in the lower bound proofs above is not closure-above itself, but
    its consequences: being a union of pseudospheres containing the full simplex,
    such that for each pseudosphere, all graphs contain the smallest graph.

    All three properties are present in a specific subset of the product of
    two simple closed-above models: all products where edges might be added to the last
    graph in the product but not to the other. Each added edge only alters
    the view of its destination, since it is in the second graph, and multiple added
    edges don't interfere because they are all added to the same graph. Hence we can
    change the views of processes one at a time, and thus we get a pseudosphere.
    Since adding no edge gives the original product and
    adding all missing edges gives the clique, we get the other two properties.
    Then taking this subset of the product of two general closed-above models result
    in a union of pseudosphres, one for each product of the underlying graphs.

    Therefore we can extract relevant subcomplexes from the product of closed-above
    models, and then the lower bounds only depends on the properties of the
    underlying product of graphs.

  \subsection{Upper bounds for multiple rounds}

    Even if we just explained how to deal with lower bounds for multiple bounds,
    we still start by giving upper bounds for multiple rounds. This is for the
    same reason as in the one round case: the upper-bounds require no combinatorial topology,
    and they allow us to introduce concepts needed for the lower bounds.

    First, we need to prove a little result that is enough for our upper bounds:
    that the product of closed-above models is included in the closure-above of the
    product.

    \begin{lemma}[Product and inclusion for closed-above]
      \label{prodIncl}
      Let $G$ and $H$ be two graphs. Then $\uparrow G
      \bigotimes \uparrow H \subseteq \uparrow~(G~\bigotimes~H)$.
    \end{lemma}

    \begin{proof}
      Let $K \in \uparrow G \bigotimes \uparrow H$. Thus
      $\exists G' \in \uparrow G, \exists H' \in \uparrow H: K=G' \bigotimes H'$.
      Let $u, v \in \Pi$ such that $(u,v) \in G \bigotimes H$. We show that
      $(u,v) \in K$; this will entail that $K \in \uparrow (G \bigotimes H)$.

      Because $(u,v) \in G \bigotimes H$, $\exists w \in \Pi:
      (u,w) \in E(G) \land (w,v) \in E(H)$. But $G' \in \uparrow G$
      and $H' \in \uparrow H$, therefore $(u,w) \in E(G') \land (w,v) \in E(H')$.
      We conclude that $(u,v) \in G' \bigotimes H' = K$.
    \end{proof}

    What this means is that taking the closure-above of the products of our graphs
    over-approximate the actual model after $r$ rounds. And thus, algorithms working
    on these approximations work on the actual model.

    Now, let us start with simple closed-above models.
    Just like for the one round case, they are completely
    characterized by the domination number of their underlying graph.

    \begin{theorem}[Upper bound (multiple rounds) for simple closed-above models]
      Let $A$ be a simple closed-above model defined by the graph $G$.
      Let $r > 0$. Then $\dom{G^r}$-set agreement is solvable
      in $r$ rounds in $A$.
    \end{theorem}

    \begin{proof}
      We have that $\dom{G^r}$-set agreement is solvable on $\uparrow G^r$ by
      Theorem~\ref{upperDomOne}. This then implies by Lemma~\ref{prodIncl}
      that it is solvable on $(\uparrow G)^r$, that is on $A$.
    \end{proof}

    But for general closed-above models, one cannot use the domination number
    itself, because one cannot know which of the underlying graphs will be
    there. As in the one round case, we use the equal-domination number and covering
    numbers.

    \begin{theorem}[Upper bound (multiple rounds) on $k$-set agreement by \eDom{S}
                    for general closed-above models]
      Let $A$ be a general closed-above model generated by the set of
      graphs $S$. Let $r > 0$.
      Then $\eDom{S^r}$-set agreement is solvable in $r$ rounds on $A$.
    \end{theorem}

    \begin{proof}
      We have that $\eDom{S^r}$-set agreement is solvable on\\
      $\bigcup\limits_{G_1,...,G_r \in S}
      \uparrow \bigotimes\limits_{i=1}^r G_i$ by
      Theorem~\ref{upperEdomOne}. This then implies by Lemma~\ref{prodIncl}
      that it is solvable on
      $\bigcup\limits_{G_1,...,G_r \in S}
      \bigotimes\limits_{i=1}^r \uparrow G_i$, that is on $A$.
    \end{proof}

    \begin{theorem}[Upper bounds (multiple rounds) on $k$-set agreement by covering numbers
                    for general closed-above models]
        \label{upperCovMulti}
        Let $A$ be a general closed-above model generated by the set of
        graphs $S$. Let $r > 0$.
        Then $\forall i \in [1, \eDom{S^r}[: (i+ (n- \cov{i}{S^r}))$-set
        agreement is solvable on the oblivious closed-above model
        generated by $S$ in $r$ rounds.
    \end{theorem}

    \begin{proof}
      We have that $\forall i \in [1, \eDom{S^r}[: (i+ (n- \cov{i}{S^r}))$-set
      agreement is solvable on
      $\bigcup\limits_{G_1,...,G_r \in S}
      \uparrow \bigotimes\limits_{i=1}^r G_i$ by
      Theorem~\ref{upperCovOne}. This then implies by Lemma~\ref{prodIncl}
      that it is solvable on
      $\bigcup\limits_{G_1,...,G_r \in S}
      \bigotimes\limits_{i=1}^r \uparrow G_i$, that is on $A$.
    \end{proof}

    One issue with these bounds is that they require the computation of
    possibly many products, as well as the computation of the combinatorial
    numbers for a lot of graphs. One alternative is to forsake the best bound
    we can get for one that can be computed using only the numbers for the initial
    graphs.

    This hinges on covering number sequences. Recall that
    the $i$-th covering number of a graph is the minimum number of processes hearing
    a set of $i$ processes that do not broadcast. In a sense, it gives the
    guaranty of propragation of information by a set of $i$ processes.

    That's the whole story for one round. But what happens when you do multiple
    rounds? Then, if the $i$-th covering number of the graph is greater than
    $i$, this means that in the next rounds, the minimum number of people
    who will hear the value of the $i$ initial processes is the $cov_i$-th covering
    number. And if this number is greater than $cov_i$, this repeats.

    Covering number sequences capture this process. One can also see them as
    the sequences of covering numbers for powers of the graph.

    \begin{definition}[Covering number sequences]
      Let $G$ be a graph. Then the \textbf{$i$-th
      covering numbers sequence} of $G$ $\triangleq
      (s^i_j)_{j \in \mathbb{N}^*}$ such that
      $s^i_1 = \cov{i}{G}$ and
      $\forall k \geq 1: s^i_{k+1} =
      \left(
      \begin{array}{ll}
        |\Pi| & \textit{if }s^i_k \geq \eDom{G}\\
        \cov{s^i_k}{G} & \textit{if }s^i_k < \eDom{G}\\
      \end{array}
      \right)$
    \end{definition}

    Armed with these sequences, we get an upper bound directly from $G$.

    \begin{theorem}[Upper bounds on $k$-set agreement by covering
            numbers sequences]
      \label{upperCovSeq}
      Let $A$ be a simple closed-above model defined by the graph $G$
      on $\Pi$. Then if the $i$-th covering
      sequence of $G$ reaches $n$ at some point, $i$-set agreement
      is solvable on the model $A$.
    \end{theorem}

%
%

    We can adapt this bound for general closed-above models by generalizing
    the covering numbers sequences to a set of graphs.

    \begin{definition}[Covering numbers sequences for sets of graphs]
      Let $s$ be set of graphs. Then the \textbf{$i$-th
      covering numbers sequence} of $S$ $\triangleq
      (s_j)_{j \in \mathbb{N}^*}$ such that
      $s_1 = \min\limits_{G \in S} cov_i(G)$ and\\
      $\forall k \geq 1: s_{k+1} =
      \left(
      \begin{array}{ll}
        n & \textit{if }s_k \geq \max\limits_{G \in S} k_{eq-dom}(G)\\
        \min\limits_{G \in S} cov_{s_k}(G) & \textit{if }s_k <
        \max\limits_{G \in S} k_{eq-dom}(G)\\
      \end{array}
      \right)$
    \end{definition}

    \begin{theorem}[Upper bounds on $k$-set agreement by covering
                    numbers sequences for general closed-above models]
      Let $S$ be a set of graph on $\Pi$. Then if the $i$-th covering
      sequence of $S$ reaches $n$ at some point, $i$-set agreement
      is solvable on the oblivious closed-above model generated by $S$.
    \end{theorem}

    \begin{proof}
      If the $i$-th covering number sequence of $S$ reaches $n$ after step $r$,
      this means that every set of $i$ processes is heard by everyone
      after $r$ rounds. In particular, the $i$ processes with the smallest
      initial values will be heard by everyone.

      Hence sending all the values heard for now for $r$ rounds, and then
      deciding the smallest value received, ensures that one of the $i$-th
      smallest values will be chosen, and thus solves $i$-set agreement.
    \end{proof}

  \subsection{Lower bounds for multiple rounds}

    \begin{theorem}[Lower bound (multiple rounds) on $k$-set agreement for
      simple closed-above models]
      \label{lowerMultiSimple}
      Let $r > 0$ and let $A$ a simple closed-above model generated by the graph
      $G$.\\
      Then $(\dom{G}-1)$-set agreement is not solvable on $A$
      in $r$ rounds by an oblivious algorithm.
    \end{theorem}

    \begin{theorem}[Lower bound (multiple rounds) on $k$-set agreement for
      general closed-above models]
      \label{lowerGeneralMultiple}
      Let $r > 0$ and let $A$ be a closed-above model generated by
      the set of graphs $S$.\\
      Let $l =
      \min(\eDomOver{S^r}-2,\min \{t+M_t(S^r)-2 \mid t \in [1,\eDomOver{S^r}-1] \})$
      Then $(l+1)$-set agreement is not solvable on $A$
      in $r$ rounds by an oblivious algorithm.
    \end{theorem}

    As a concrete applications of these bounds, we consider a classical family
    of subgraphs: stars.

    \begin{definition}[Star graphs]
      Let $G$ be a graph. Then $G$ is a \textbf{star graph}
      $\triangleq \exists S \subseteq \Pi: G = (V,S \times \Pi)$.
    \end{definition}

    \begin{theorem}[Lower bound for stars]
      \label{lowerStar}
      Let $S$ be the set of graphs which are unions of $s$ stars with different
      centers. Then $n-s$-set agreement is
      not solvable in the closed-above model generated by $S$.
    \end{theorem}

\section{Conclusion}
\label{sec:conclu}

  We provided upper and lower bounds on $k$-set agreement for closed-above
  models, the subset of round-based models defined by subgraphs that
  must be present in the communication graph at each round.
  These models encompass many message-passing models of distributed computing
  focused on safety properties.

  Regarding the bounds themselves, although their proofs leverage
  combinatorial topology, all our bounds
  are expressed in terms of combinatorial numbers of the graphs. That is,
  these bounds can be used without any knowledge of combinatorial topology.
  Yet combinatorial topology was instrumental in showing such sweeping results.

\begin{acks}
  Adam Shimi was supported by the
  \grantsponsor{ANR}{Agence Nationale de Recherche}{https://anr.fr}
  under Grant No.:~\grantnum{ANR}{PARDI ANR-16-CE25-0006} and
  Armando Casta\~neda was supported by project PAPIIT IN108720.
\end{acks}

\bibliography{references}

\appendix

\section{Proof of Lemma~\ref{lemma-shellability-connectivity}}

  \begin{proof}
      We prove the claim by induction on $\ell$.
      \begin{itemize}
          \item \textbf{(Base case)} $\ell=0$.
              We need to prove that $\cB$ is $0$-connected,
              by induction on the length of a shelling order of $\cA$.
              Fix a shelling order $\phi_1, \ldots, \phi_m$ of $\cA$,
              so $\cB = \bigcup^m_{i=1} \alpha(\phi_i)$.
              \begin{itemize}
                  \item \textbf{(Base case)} $\cB=\alpha(\phi_1)$. By hypothesis
                  \textbf{(2)}, for the case $t=0$, $\cB=\alpha(\phi_1)$ is
                  at least $l-t = 0$-connected.
                  \item \textbf{(Induction step)} Suppose that
                      $\bigcup^{r-1}_{i=1} \alpha(\phi_i)$ is $0$-connected,
                      for some $2 \leq r < m$.
                      We have that $\alpha(\phi_r)$ is $0$-connected by hypothesis
                      \textbf{(2)}, as above.
                      We show that $\left( \bigcup^{r-1}_{i=1} \alpha(\phi_i)
                      \right) \cap \alpha(\phi_r)$ is $(-1)$-connected,
                      namely, non-empty, and then Corollary~\ref{simpleNerve} imply that
                      $\cB=\left(\bigcup^{r-1}_{i=1}
                      \alpha(\phi_i)\right)\cup\alpha(\phi_r)$ is $0$-connected.
                      By definition of shellability,
                      \[
                      \left( \bigcup^{r-1}_{i=1} \phi_i \right) \cap \phi_r =
                      \tau_1 \cup \ldots \cup \tau_s,
                      \]
                      where each $\tau_j$ is a face of dimension $(d-1)$ of $\phi_r$.
                      For each $\tau_j$ there is a facet $\sigma_j$ of
                      $\bigcup^{r-1}_{i=1} \phi_i$
                      such that $\tau_j \subset \sigma_j$.
                      Thus, $\phi_r$ and $\sigma_j$ share a $(d-1)$-face and
                      \[
                      \left( \bigcup^{r-1}_{i=1} \phi_i \right) \cap \phi_r =
                      \bigcup^{s}_{j=1} \left( \phi_r \cap \sigma_j \right).
                      \]
                      By hypothesis (1) we have that
                      \[
                      \left( \bigcup^{r-1}_{i=1} \alpha(\phi_i) \right)
                      \cap \alpha(\phi_r) = \bigcup^{s}_{j=1} \left( \alpha(\phi_r)
                      \cap \alpha(\sigma_j) \right).
                      \]
                      Each $\sigma_j$ shares a $(d-1)$-face with $\phi_r$,
                      so hypothesis~(2), with $t=1$, implies that
                      $\alpha(\phi_r) \cap \alpha(\sigma_j)$ is at least
                      $(-1)$-connected,
                      which implies that
                      $\left( \bigcup^{r-1}_{i=1} \alpha(\phi_i) \right)
                      \cap \alpha(\phi_r)$ is non-empty.
              \end{itemize}
          \item \textbf{(Induction step)} Suppose that
              the statement of the theorem holds for $\ell-1$,
              and consider a shelling order $\phi_1, \ldots, \phi_m$ of $\cA$.
              Our aim is to show that ${\cB} = \bigcup^m_{i=1} \alpha(\phi_i)$
              is $\ell$-connected.

              As in the base case, we proceed by induction on the
              length of the shelling order.
              \begin{itemize}
                  \item \textbf{(Base case)} $\cB=\alpha(\phi_1)$.
                  By hypothesis \textbf{(2)}, for the case $t=0$, $\cB=\alpha(\phi_1)$ is
                  at least $\ell-0 = \ell$-connected.
                  \item \textbf{(Induction step)} Suppose that
                      $\bigcup^{r-1}_{i=1} \alpha(\phi_i)$ is
                      $\ell$-connected, for some $2 \leq r < m$.

                      We have that $\alpha(\phi_r)$ is $\ell$-connected by hypothesis
                      \textbf{(2)} as above.
                      If we show that
                      $\left( \bigcup^{r-1}_{i=1} \alpha(\phi_i) \right)
                      \cap \alpha(\phi_r)$ is $(\ell-1)$-connected,
                      Corollary~\ref{simpleNerve} implies that
                      $\bigcup^r_{i=1} \alpha(\phi_i)$ is $\ell$-connected.

                      To do so, we use the theorem for $\ell-1$.
                      As seen before, there are facets
                      $\sigma_1, \ldots, \sigma_s$ of
                      $\bigcup^{r-1}_{i=1} \phi_i$ such that
                      each $\sigma_j$ and $\phi_r$ share a $(d-1)$-face,
                      \[\left(
                      \bigcup^{r-1}_{i=1} \phi_i \right) \cap \phi_r =
                      \bigcup^{s}_{j=1} \left( \phi_r \cap \sigma_j
                      \right) \; \mbox{and} \;\; \left(
                      \bigcup^{r-1}_{i=1} \alpha(\phi_i) \right) \cap
                      \alpha(\phi_r) =
                      \bigcup^{s}_{j=1} \left( \alpha(\phi_r) \cap
                      \alpha(\sigma_j) \right).
                      \]

                      Let $\cB' = \bigcup^{s}_{j=1} \left( \alpha(\phi_r)
                      \cap \alpha(\sigma_j) \right)$.
                      Let $\lambda_1, \hdots, \lambda_{s'}$ be simplexes
                      among the $\sigma_j$'s such that
                      $\alpha(\phi_r)
                      \cap \alpha(\lambda_i) \neq \alpha(\phi_r) \cap
                      \alpha(\lambda_{i'})$ for $i\neq i'$,
                      and the
                      $\alpha(\phi_r) \cap \alpha(\lambda_i)$ still
                      form a cover of $\cB'$: $\cB' = \bigcup^{s'}_{i=1}
                      \left( \alpha(\phi_r) \cap \alpha(\lambda_i) \right)$.
                      Let $\cA' = \bigcup^{s'}_{i=1}
                      \left( \phi_r \cap \lambda_i \right)$.
                      Note that $\cA'$ is pure of dimension $d-1$ and
                      is a subcomplex of the boundary complex of $\phi_r$.

                      By Lemma~\ref{lemma-shelling-boundary},
                      $\cA'$ is shellable.
                      The facets of $\cA'$ are the intersections
                      $\phi_r \cap \lambda_i$.
                      Consider the bijection $\beta(\phi_r \cap \lambda_i) =
                      \alpha(\phi_r) \cap \alpha(\lambda_i)$
                      between the facets of $\cA'$ and our cover of $\cB'$.
                      Now, let $\phi_r \cap \lambda$ be any facet of $\cA'$
                      and $\bigcup^{m'}_{i=1} \left( \phi_r \cap \lambda'_i \right)$
                      be any pure $(d-1)$-subcomplex of $\cA'$.
                      Note that every pair of facets of $\cA'$ share
                      a $(d-2)$-face as both are $(d-1)$-faces of $\phi_r$.
                      Then, $\phi_r \cap \lambda$ and
                      each $\phi_r \cap \lambda'_i$ share a
                      face of dimension $d-2$, and thus we can write

                      $$\left( \bigcup^{m'}_{i=1} \left( \phi_r \cap
                      \lambda'_i \right) \right) \cap \left( \phi_r \cap
                      \lambda \right) = \bigcup^{m'}_{i=1}
                      \left( (\phi_r \cap \lambda) \cap
                      (\phi_r \cap \lambda'_i) \right)$$
                      and
                      $$\left( \bigcup^{m'}_{i=1} \beta(\phi_r \cap \lambda'_i)
                      \right) \cap \beta(\phi_r \cap \lambda) =
                      \bigcup^{m'}_{i=1} \\	\left( \beta(\phi_r \cap \lambda)
                      \cap \beta(\phi_r \cap \lambda'_i) \right).$$
                      We conclude that hypothesis (1) of the theorem
                      holds for $\cA'$, $\cB'$ and $\beta$.

                      Finally, consider any collection
                      $\phi_r \cap \lambda'_0, \ldots,
                      \phi_r \cap \lambda'_{t'}$ of $t'+1$ facets of $\cA'$.
                      As already noted, each of them and
                      the first one share a $(d-2)$-face.
                      We have that
                      $$\tau = \bigcap^{t'}_{i=0} \beta(\phi_r \cap \lambda'_i)
                      = \bigcap^{t'}_{i=0} (\alpha(\phi_r) \cap
                      \alpha(\lambda'_i)) = \alpha(\phi_r) \cap
                      \bigcap^{t'}_{i=0} \alpha(\lambda'_i).$$
                      As said above, the $\lambda'_i$'s are facets of
                      $\cA$ and each of them and $\phi_r$ share a $(d-1)$-face.
                      By hypothesis (2) with $t=t'+1$, $\tau$ is
                      of at least $(\ell - t) =
                      (\ell - (t'+1)) = ((\ell - 1) - t')$-connected.
                      Then, hypothesis (2) of the theorem holds
                      for $\cA'$, $\cB'$, $\beta$ and $\ell-1$.

                      We have all hypothesis to use the theorem
                      with $\cA'$ and $\cB'$ and $\ell-1$.
                      Therefore, $\cB'$ is $(\ell-1)$-connected,
                      and then $\cup^{r}_{i=1} \alpha(\phi_i)$
                      is $\ell$-connected.
              \end{itemize}
      \end{itemize}
  \end{proof}

\section{Proof of Theorem~\ref{lowerGeneral}}

  \begin{proof}
    It is known that when the protocol complex is $k$ connected, non trivial
    $k+1$-set agreement is impossible. Herlihy et al.~\cite{HerlihyBook} gives
    an example derivation for colorless protocols,
    and Casta{\~n}eda et al.~\cite{Sirocco} give one fore colored protocols.
    We thus prove that the protocol complex generated by $A$ after one round is
    $l$-connected.

    As we said, we want to apply Lemma~\ref{lemma-shellability-connectivity}.
    Our $\cA$ is the pseudosphere $\Psi(\Pi,[0,k])$, our $\cB$
    is $C_A(\cA)$ and our mapping $\alpha$ sends a facet $\sigma$
    of $\cA$ on $C_A(\sigma)=\bigcup\limits_{G \in S} C_G(\sigma)$.
    \begin{itemize}
      \item Let $\phi'$ be a facet of $\cA$ and take a
        pure $d$-subcomplex $\bigcup^t_{i=1} \phi_i
        \subseteq \cA$ satisfying that
        $\left( \bigcup^t_{i=1} \phi_i \right) \cap \phi' =
        \bigcup^s_{i=1} \left( \sigma_i \cap \phi' \right)$
        for some of $\cA$'s facets $\sigma_1, \hdots, \sigma_s$,
        with each $\sigma_i$ and $\phi'$ sharing a $(d-1)$-face.

        We want to show that
        $\left( \bigcup^t_{i=1} \alpha(\phi_i) \right)
        \cap \alpha(\phi') = \bigcup^s_{i=1}
        \left( \alpha(\sigma_i) \cap \alpha(\phi') \right)$.
        \begin{itemize}
          \item Let $\tau$ be a simplex of
            $\left( \bigcup^t_{i=1} \alpha(\phi_i) \right)
            \cap \alpha(\phi')$.
            Since $\left( \bigcup^t_{i=1} \alpha(\phi_i) \right)
            \cap \alpha(\phi') =\\
            \bigcup^t_{i=1}
            \left( \alpha(\phi_i) \cap \alpha(\phi') \right)$
            by distributivity of intersection on union,
            we have some $i \in [1,t]$ such that $\tau$ is
            a simplex of $\alpha(\phi_i) \cap \alpha(\phi')$.

            Now, $\alpha$ sends a simplex $\sigma$ to $C_A(\sigma)$;
            thus $\alpha(\phi') \cap \alpha(\phi_i) =
            C_A(\phi') \cap C_A(\phi_i) =
            (\bigcup\limits_{G \in S} C_G(\phi')) \cap
            (\bigcup\limits_{G \in S} C_G(\phi_i)) =
            \bigcup\limits_{G,H \in S} C_G(\phi') \cap C_H(\sigma_i)$.
            Hence $\tau$ being a simplex of
            $\alpha(\phi_i) \cap \alpha(\phi')$ means that $\exists G,H \in S$
            for which every process of $\tau$ has its view in both
            $C_G(\phi_i)$ and $C_H(\phi')$. And a view is completely
            defined by the value received from other processes.
            That is, $\forall (p,v) \in \tau, \forall (q,v_q) \in v:
            (q,v_q) \in \phi_i \cap \phi'$.

            By our equation
            $\left( \bigcup^t_{i=1} \phi_i \right) \cap \phi' =
            \bigcup^s_{i=1} \left( \sigma_i \cap \phi' \right)$
            there is $l \in [1,s]$ such that
            $\phi_i \cap \phi' \subseteq \sigma_l \cap \phi'$.
            Then all $(q,v_q) \in v$ are also in
            $\sigma_l \cap \phi'$.
            We conclude that all $(q,v_q) \in v$
            are in $\sigma_l$ and in $\phi'$,
            and thus $\tau$ is a simplex of
            $\alpha(\sigma_l) \cap \alpha(\phi')$.
          \item Let $\tau$ be a simplex of
            $\bigcup^s_{i=1} \left( \alpha(\sigma_i)
            \cap \alpha(\phi') \right)$. Thus there is
            some $i \in [1,s]$ such that $\tau$ is a simplex of
            $\alpha(\sigma_i) \cap \alpha(\phi')$.
            That is, $\forall (p,v) \in \tau, \forall (q,v_q) \in v:
            (q,v_q) \in \sigma_i \cap \phi'$.

            Then, by our equation
            $\left( \bigcup^t_{i=1} \phi_i \right) \cap \phi' =
            \bigcup^s_{i=1} \left( \sigma_i \cap \phi' \right)$,
            there is $l \in [1,t]$ such that $\forall (p,v) \in \tau,
            \forall (q,v_q) \in v: (q,v_q) \in \phi_l$.
            We conclude that all $(q,v_q) \in v$ are in $\phi_l$
            and in $\phi'$, and thus $\tau$ is a simplex of
            $\left( \bigcup^t_{i=1} \alpha(\phi_i) \right)
            \cap \alpha(\phi')$.
        \end{itemize}
      \item Now we want to study the connectivity of the
        intersection of well-chosen facets.
        Let $t\geq 0$ and $\phi_0, \phi_1, \ldots, \phi_t$ be
        $t+1$ facets of $\mathcal{A}$ with each $\phi_i$ and $\phi_0$
        sharing a $(d-1)$-face.
        We want to prove that $\bigcap^t_{i=0} \alpha(\phi_i)$ is $l-t$
        connected.

        Because $l \leq \eDomOver{S} - 2$, we only have to consider
        $t < \eDomOver{S}$, because $l-\eDomOver{S} \leq -2$, and thus
        in the case $t \geq \eDomOver{S}$, there is no constraint
        to satisfy on the connectivity of the intersection.

        Now, each $\alpha(\phi_i)$ is in fact $C_A(\phi_i)
        = \bigcup\limits_{G \in S} C_G(\phi_i)$.
        We start by developping the $C_A(\phi_i)$ into the
        union of the $C_G(\phi_i)$ and applying the distributivity
        of intersection over union on this big intersection:
        \[
        \begin{array}{ll}
          \bigcap\limits_{i \in [0,t]} C_A(\phi_i) & =
            \bigcap\limits_{i \in [0,t]}
            \bigcup\limits_{G \in S} C_G(\phi_i)\\
          & = \bigcup\limits_{G_0,G_1,...,G_t \in S}
            \bigcap\limits_{i \in [0,t]}
              C_{G_i}(\phi_i)\\
        \end{array}
        \]
        As always, we naturally get a cover of our space.
        We thus use the Nerve Lemma.

        This requires first a computation of connectivity for the
        $\bigcap^t_{i=0} C_{G_i}(\phi_i)$. We are taking
        the intersection of pseudospheres, which gives a new
        pseudosphere by Lemma~\ref{capPseudo}. To compute its
        connectivity, we need to now how much processes end up
        with a non-empty set of view, by Lemma~\ref{connPseudo}.

        Let us assume that the $\phi_i$ are all distinct; if not we
        can remove the duplicate and start with a lower $t$.
        Then, because they all intersect with $\phi_0$ on
        a $(d-1)$ face, we have $\bigcap^t_{i=0} \phi_i$ of
        dim $(d-t)$. That is, in these input simplexes, there
        are $(d-t)$ processes with the same input value across
        all $\phi_i$. Or equivalently, there are $t$ processes
        with different values for some $\phi_i$.

        Let $P$ be the set of $t$ processes with sometimes
        different initial values across the $\phi_i$. Then
        the processes disappearing from
        $\bigcap^t_{i=0} C_{G_i}(\phi_i)$ are the ones
        receiving the values from $P$.

        But for $t < \eDomOver{S}$,
        we know either all processes receive the values from $P$,
        or at most $\mCov{t}{S}$ do.
        Thus $\bigcap^t_{i=0} C_{G_i}(\phi_i)$ is either empty or
        a pseudosphere with $(n-\mCov{t}{S})$ processes with a non empty
        set of views.
        It is therefore empty or $(n-\mCov{t}{S}-2)$-connected
        by Lemma~\ref{connPseudo}.

        Let us index the subsets of $S$ of size $t$ whose intersection
        is non-empty.
        Then let $J$ be a set of index of size $\leq M_t(S)$. We now
        show that $\bigcap\limits_{j \in J}
        \bigcap\limits_{\substack{i \in [0,t] \\G_i \in S_j}} C_{G_i}(\phi_i)$
        is either empty or $(M_t(S)-2 - |J| + 1) = (M_t(S) - |J| -1)$-connected.

        If $P$ dominates any set $S_j$ for $j \in J$, then the intersection
        for this set is empty, and thus the intersection of the intersections
        of $J$ is also empty. Let us thus assume from this point that $P$
        does not dominate any set $S_j$ with $j \in J$.

        Then in each $G \in S_j$, $P$ talks to at most \mCov{t}{S} processes. Of
        these, there are\\
        $\mCov{t}{S} - t$ who are not in $P$. This
        means that in the worst case, $P$ talks to $|J|(\mCov{t}{S} - t)$
        processes not in $P$. Thus in this worst case, the number
        of processes hearing the views of P is $t + |J|(\mCov{t}{S}-t)$.

        Therefore, the intersection is $(n-t-|J|(\mCov{t}{S}-t)-2)$-connected.

        First we treat the case where $\mCov{t}{S} = t$, and thus where
        $M_t(S) = n-t$. We have an intersection that is $(n-t-2)$-connected, and
        $n-t -2 \geq n-t - |J| -1$, since $|J| \geq 1$.
        This actually holds for any J, even when $|J| > M_t(S)$. Hence the
        nerve complex of our cover in this case is a simplex,
        and thus $\infty$-connected. We conclude by
        the nerve lemma~\ref{nerveLemma} that $\bigcap^t_{i=0} \alpha(\phi_i)$
        is $(M_t(S) - 2)$-connected.

        Now let us turn to the case where $\mCov{t}{S} > t$.
        Notice that $n-t-|J|(\mCov{t}{S}-t)-2 = (n-t-1) - |J|(\mCov{t}{S}-t)-1
        \geq (\mCov{t}{S}-t)(M_t(S)-|J|) - 1$. And since $M_t(S) \geq |J|$ and
        $\mCov{t}{S} > t$, we have $n-t-|J|(\mCov{t}{S}-t)-2 \geq M_t(S)-|J|-1$.
        Among other things, this means that $\forall J$ such that $|J| = M_t(S)$,
        $\bigcap\limits_{j \in J}
        \bigcap\limits_{\substack{i \in [0,t] \\G_i \in S_j}} C_{G_i}(\phi_i)$
        is $(-1)$-connected and thus not empty.

        This implies that the nerve complex contains the $M_t(S)$ skeleton
        of a higher dimensional simplex. And such a skeleton is at least
        $(M_t(S)-1)$-connected. We conclude by the nerve lemma~\ref{nerveLemma}
        that $\bigcap^t_{i=0} \alpha(\phi_i)$ is $(M_t(S)-2)$-connected.

        For $t < \eDomOver{S}$, we thus have that
        $\bigcap^t_{i=0} \alpha(\phi_i)$ is
        $(M_t(S)-2)$-connected. And since
        $(t + M_t(S)-2) \geq l$, we have $M_t(S)-2 \geq l-t$, and thus
        we are $(l-t)$-connected.
    \end{itemize}
  \end{proof}

\section{Proof of Corollary~\ref{lowerSym}}

  \begin{proof}
    We simply compute $M^t(Sym(G))$ from $M_t(\{G\})$.
    First, notice that if $\mCov{t}{\{G\}} = t$, then no other process hears
    any set of $t$ processes. This is invariant by permutation, and thus
    $\mCov{t}{S} = t$, and $M_t(Sym(G)) = n-t$.

    We now turn to the case $\mCov{t}{\{G\}} = t$.
    In the worst case, we have a $P \subseteq \Pi$ of size $t$ that hits
    \mCov{t}{\{G\}} processes in $G$. Of these, $\mCov{t}{\{G\}} - t$ processes
    are not in $P$. By permutation, we can take, in the worst case,
    $t-1$ other graphs where these $\mCov{t}{\{G\}}-t$ are completely different.
    That is, in the
    worst case, $P$ touches $\mCov{t}{Sym(G)} = t + t(\mCov{t}{\{G\}}-t)$ processes. And
    thus $M_t(Sym(G) = \left\lfloor \frac{n-t-1}{\mCov{t}{Sym(G)}-t} \right\rfloor
    = \left\lfloor \frac{n-t-1}{t(\mCov{t}{G}-t)} \right\rfloor$.
  \end{proof}

\section{Proof of Theorem~\ref{upperCovSeq}}

  \begin{proof}
    Let $S$ be the set of the $k$ smallest
    initial values. There are at least $k$
    processes with one of these values.
    Now the $k$-th covering sequence
    of $G$ gives us a lower bound on
    the number of processes
    who know these values for each round.
    We show this by induction on the index $j$
    of the sequence elements.
    \begin{itemize}
      \item \textbf{(Base case)} $j = 1$.
        Then after one round, all processes
        who heard from our $k$ initial processes
        know their initial values; this number is lower
        bounded by $cov_k(G)$.
      \item \textbf{(Induction step)} $j = t+1$ and
        the result holds $\forall j \leq t$. Notably,
        $s_t^k$ lower bounds the number of processes
        knowing one of the $k$ initial values after
        $t$ rounds.
        \begin{itemize}
          \item if $s_t \geq \eDom{G}$ then whatever
            the set of processes knowing one of the $k$
            values after round $t$, they form a dominating
            set. And thus every process will know at least
            on of these values after round $t+1$, corresponding
            to $s_{t+1}^k = n$.
          \item If $s_t^k < \eDom{G}$, then not all sets
            of size $s_t^k$ are dominating sets. But they
            all reach at least $cov_{s_t^k}(G)$ processes.
            Thus after round $t+1$, at least that much
            processes known at least one of the $k$
            initial values.
        \end{itemize}
    \end{itemize}

    Hence, if the $k$-th covering sequence reaches
    $n$ after say round $t$, all processes know at least
    one initial value from every set of $k$ processes. Notably,
    every process knows one of the $k$ smallest initial values,
    and thus choosing the smallest value ensure that
    $k$-set agreement is solved.

    Thus the algorithm where every one send their view for $t$
    rounds, and then decide the minimum value they know, solves
    $k$-set agreement on $G$.
  \end{proof}

\section{Proof of Theorem~\ref{lowerMultiSimple}}

  \begin{proof}
    Because we only consider oblivious algorithms, we can consider the graphs
    given by the products of $r$ graphs of $A$ as generating a model, and apply
    our bound for one round. The trick is to know if $(\uparrow G)^r$ is itself
    a simple closed-above model. One would think that probably
    $(\uparrow G)^r = \uparrow G^r$, but this is not true in general, as shown
    in the examples at the beginning of the section.

    On the other hand, $(\uparrow G)^r$ contains a subcomplex which has
    all the properties that we use in the proof of our lower bound:
    $G^{r-1} . \uparrow G$.
    \begin{itemize}
      \item It is a pseudosphere. Indeed, such a complex contains
        the uninterpreted complex of $G^{r-1} . G = G^r$; and each
        edge added to $G$ only changes the view of
        one process in the product, the one receiving the
        new message.

        Hence we can change the view of each process independantly
        of the others (because we only add messages to the
        last graph, and thus no two such messages can interfere
        with each other to create a new path).
      \item It contains the full simplex. This follows from the
        fact that $\uparrow G$ contains the clique, and our
        product operation maintain this graph.
      \item It is included in the complex $\uparrow G^r$. This follows from the
        fact that adding messages to $G$ can only add messages to
        the product, and thus all graphs considered contain
        $G^r$ and thus are in $\uparrow G^r$.
    \end{itemize}

    This means that this subcomplex can be treated just like
    the complex of $\uparrow G^r$ in our lower bound for one round. The
    theorem follows.
  \end{proof}

\section{Proof of Theorem~\ref{lowerGeneralMultiple}}

  \begin{proof}
    The idea is the same that for the proof above, except that for
    each product of r graphs $G_1 .G_2 . \dots . G_r$, we consider
    the complex of the graphs in $G_1 . G_2 . \dots . G_{r-1} . \uparrow G_r$.

    This satisfies the same three properties used in our lower bound proof
    than $\uparrow (G_1 .G_2 . \dots . G_r)$, and thus our lower bound
    gives the same result.

    This means we can consider the union of our complexes like the
    the union of the complexes for $\uparrow (G_1 .G_2 . \dots . G_r)$ for
    each product of $r$ graphs of $S$, that is like the complex of the
    general closed-above model generated by $S^n$.
  \end{proof}

\section{Proof of Theorem~\ref{lowerStar}}

  \begin{proof}
    First, we have $\forall r > 0 : \eDomOver{S^r} = \eDomOver{S} = n-s+1$.

    The first equality follows from the fact that any star graph is idempotent
    by ouf product; hence $S \subseteq S^r$, and a set of $n-s$ procs in the
    $n-s$ graphs with stars at the other $s$ procs does not dominate these
    graphs. This actually gives $\eDomOver{S^r} \geq \eDomOver{S}$; the other
    direction follows form the fact that $\forall G \in S^r, \exists H \in S:
    G \in \uparrow H$. Hence every set of $i$ graphs in $S^r$ will have more
    edges than some set of $i$ graphs in $S$, and thus will be easier to dominate.

    As for the second equality, it follows from the fact that with $n-s+1$ processes
    and $n-s+1$ distinct graphs, at least one of the process is the center of the
    star in one of the graph, and thus the processes dominates the graphs.

    On the other hand, $\forall t \in [1,\eDomOver{S}-1]:
    t + M^t(S^r)-2=n-2$. This equality comes from the fact
    that any set of $t$ processes that does not dominate a given set of $t$ graphs of
    $S^r$ is distinct from the stars of these $t$ graphs. Thus they are silent, and
    then $\mCov{t}{S^r} = t$ and $M^t(S^r) = n-t$. We thus have
    $t + M^t(S^r)-2 = t+n-t-2 = n-2$.

    Therefore
    $\min(\eDomOver{S^r}-2,\min \{t+M_t(S^r)-2 \mid t \in [1,\eDomOver{S^r}-1] \})
    = \min(n-s+1-2,n-2) = n-s-1$. We conclude by Theorem~\ref{lowerGeneralMultiple}.
  \end{proof}

\end{document}